\documentclass[11pt,draftclsnofoot,peerreviewca,onecolumn,letterpaper]{IEEEtran}
\usepackage{amsmath,amssymb,dsfont,stfloats,color,url}
\usepackage[pdftex]{graphicx}
\usepackage{subfigure}
\usepackage{cite}
\newtheorem{thm}{Theorem}

\newtheorem{lem}{Lemma}
\newtheorem{cor}{Corollary}

\newtheorem{rem}{Remark}

\newtheorem{fact}{Fact}

\def\pd{{\pmb d}}\def\ph{{\pmb h}}\def\px{{\pmb x}}\def\py{{\pmb y}}\def\pz{{\pmb z}}

\def\pB{{\pmb B}}\def\pH{{\pmb H}}\def\pI{{\pmb I}}\def\pP{{\pmb P}}\def\pR{{\pmb R}}\def\pU{{\pmb U}}\def\p0{{\pmb 0}}

\long\def\comment#1{}

\newfont{\bbb}{msbm10 scaled 700}

\newfont{\bbc}{msbm10 scaled 1100}
\newcommand{\CC}{\mbox{\bbc C}}

\newcommand{\RR}{\mbox{\bbc R}}

\newcommand{\EE}{\mbox{\bbc E}}

\newcommand{\dv}{{\pmb d}}
\newcommand{\ev}{{\pmb e}}
\newcommand{\fv}{{\pmb f}}

\newcommand{\hv}{{\pmb h}}

\newcommand{\kv}{{\pmb k}}

\newcommand{\nv}{{\pmb n}}

\newcommand{\qv}{{\pmb q}}
\newcommand{\rv}{{\pmb r}}

\newcommand{\uv}{{\pmb u}}
\newcommand{\vv}{{\pmb v}}
\newcommand{\wv}{{\pmb w}}
\newcommand{\xv}{{\pmb x}}
\newcommand{\yv}{{\pmb y}}
\newcommand{\zv}{{\pmb z}}
\newcommand{\zerov}{{\pmb 0}}

\newcommand{\Am}{{\pmb A}}
\newcommand{\Bm}{{\pmb B}}
\newcommand{\Cm}{{\pmb C}}
\newcommand{\Dm}{{\pmb D}}
\newcommand{\Em}{{\pmb E}}
\newcommand{\Fm}{{\pmb F}}
\newcommand{\Gm}{{\pmb G}}
\newcommand{\Hm}{{\pmb H}}
\newcommand{\Id}{{\pmb I}}
\newcommand{\Jm}{{\pmb J}}
\newcommand{\Km}{{\pmb K}}

\newcommand{\Mm}{{\pmb M}}

\newcommand{\Om}{{\pmb O}}
\newcommand{\Pm}{{\pmb P}}

\newcommand{\Rm}{{\pmb R}}
\newcommand{\Sm}{{\pmb S}}
\newcommand{\Tm}{{\pmb T}}
\newcommand{\Um}{{\pmb U}}
\newcommand{\Vm}{{\pmb V}}
\newcommand{\Wm}{{\pmb W}}
\newcommand{\Xm}{{\pmb X}}
\newcommand{\Ym}{{\pmb Y}}
\newcommand{\Zm}{{\pmb Z}}

\newcommand{\Ac}{{\cal A}}

\newcommand{\Cc}{{\cal C}}

\newcommand{\Ic}{{\cal I}}
\newcommand{\Jc}{{\cal J}}
\newcommand{\Kc}{{\cal K}}

\newcommand{\Nc}{{\cal N}}

\newcommand{\Rc}{{\cal R}}
\newcommand{\Sc}{{\cal S}}

\newcommand{\muv}{\hbox{\boldmath$\mu$}}

\newcommand{\Lambdam}{\hbox{\boldmath$\Lambda$}}

\newcommand{\Sigmam}{\hbox{\boldmath$\Sigma$}}
\newcommand{\Phim}{\hbox{\boldmath$\Phi$}}

\newcommand{\Xim}{\hbox{\boldmath$\Xi$}}

\newcommand{\trace}{{\hbox{tr}}}

\newcommand{\SNR}{{\sf SNR}}

\renewcommand{\Re}{{\rm Re}}

\newcommand{\herm}{{\sf H}}

\newcommand{\transp}{{\sf T}}

\IEEEoverridecommandlockouts

\begin{document}

\title{Joint Spatial Division and Multiplexing}

\author{\authorblockN{Ansuman Adhikary\authorrefmark{2},
Junyoung Nam\authorrefmark{1}, Jae-Young Ahn\authorrefmark{1}, and Giuseppe Caire\authorrefmark{2}}
\thanks{
$^*$ Mobile Communications Division, Electronics Telecommunications Research Institute, Daejeon, Korea.

$^\dagger$ Ming-Hsieh Department of Electrical Engineering, University of Southern California, CA.

This work was supported by the IT R\&D program of MKE/KEIT in Korea [Development of beyond 4G technologies for smart mobile services].}}

\maketitle

\begin{abstract}
We propose Joint Spatial Division and Multiplexing (JSDM), an
approach to multiuser MIMO downlink that exploits the structure of
the correlation of the channel vectors
in order to allow for a large number of antennas at the
base station while requiring reduced-dimensional Channel State Information at the Transmitter (CSIT).
This allows for significant savings both in the downlink training and in the CSIT feedback from
the user terminals to the base station, thus making the use of a large number of base station antennas
potentially suitable also for Frequency Division Duplexing (FDD) systems, for which
uplink/downlink channel reciprocity cannot be exploited.
JSDM forms the multiuser MIMO downlink precoder by concatenating
a {\em pre-beamforming} matrix, which depends only on the channel second-order statistics,
with a classical multiuser precoder, based on the instantaneous knowledge
of the resulting reduced dimensional ``effective'' channels.
We prove a simple condition under which JSDM incurs no loss of optimality with respect
to the full CSIT case. For linear uniformly spaced arrays,  we show that such condition is closely approached when the number of antennas is large.
For this case, we use Szego's asymptotic theory of large Toeplitz matrices to design a DFT-based pre-beamforming scheme
requiring only coarse information about the users angles of arrival and angular spread.
Finally, we extend these ideas to the case of a two-dimensional base station antenna array,
with 3-dimensional beamforming, including multiple beams in the elevation angle direction.
We provide guidelines for the pre-beamforming optimization and calculate the system spectral efficiency under
proportional fairness and max-min fairness criteria,  showing extremely attractive performance.
Our numerical results are obtained via an asymptotic random matrix theory tool known as ``deterministic equivalent'' approximation,
which allows to avoid lengthy Monte Carlo simulations and provide accurate results
for realistic (finite) number of antennas and users.
\end{abstract}


\newpage

\section{Introduction}
\label{sec:intro}

In a Multiuser MIMO (MU-MIMO) downlink where
a base station (BS) with $M$ antennas serves $K$ single-antenna user
terminals (UTs) on the same time-frequency slot, and the
channel fading coefficients can be considered constant over coherence blocks of
$T$ channel uses,~\footnote{A channel use corresponds to an independent complex signal-space dimension in the time-frequency domain.}
the high-SNR system spectral efficiency behaves at best as $M^\star(1 -
M^\star/T) \log \SNR + O(1)$, where $M^\star = \min\{M,K,T/2\}$.
The upper bound yielding this behavior is obtained by letting
all UTs cooperate and using the result of \cite{zheng2002communication} on the high-SNR capacity of the non-coherent
block-fading MIMO point-to-point channel. A tight lower bound is obtained
by devoting $M^\star$ dimensions per block to training, in order to acquire the Channel State Information
at the Transmitter (CSIT), i.e., to estimate the downlink channel matrix on each fading
coherence block.
In Frequency Division Duplexing (FDD) systems, where the fading channel reciprocity cannot be exploited,
the lower bound is achievable by assuming ideal instantaneous CSIT feedback from the UTs to the BS,
between the downlink training phase and the data transmission phase.
If, more realistically, instantaneous feedback in the same fading coherence block is not possible,
a prediction error further decreases the system multiplexing gain by the factor $\max\{1 - 2 B_d T_s,0\}$,
where $B_d = v f_0/c$ is the Doppler bandwidth (Hz), ($v$ denoting the UT speed
(m/s), $f_0$ the carrier frequency (Hz) and $c$ the light speed
(m/s)), and $T_s$ is the slot duration (s) \cite{Caire-Jindal-Kobayashi-Ravindran-TIT10,kobayashi2011training}.

It is evident that, even not taking into account the cost of CSIT feedback (which may impact the
uplink system capacity), the MU-MIMO multiplexing gain for an FDD system based
on downlink training, channel estimation (and possibly prediction) at the UTs,
and CSIT feedback, is significantly reduced when $M^\star$ is not much smaller
than  $T$ and/or $2 B_d T_s$ is not much smaller than 1.
In particular, for large $M$ and $K$, the downlink training represent a significant bottleneck
(as quantified by the analysis in \cite{Huh-Tulino-Caire-TITsubmit}) and
the corresponding CSIT feedback yields an unacceptably high overhead for the uplink.

Alternatives that do not require CSIT \cite{gou2010aiming}  or require only outdated CSIT
\cite{maddah2010completely} (without requiring a strict one-slot prediction constraint)
have been proposed. Although these schemes may achieve better multiplexing gain
than the basic training and feedback scheme in certain conditions
(see for example the comparison in \cite{caire2012isit})
they do not scale well with the number of BS antennas and UTs,
since they require a precoding block length (in time slots) that grows very rapidly with the
number of system antennas.~\footnote{For example, \cite{maddah2010completely}
requires precoding over $M! \sum_{j=1}^M
\frac{1}{j}$ time slots in order to serve $M$ UTs with
$M$ BS antennas.}
Hence, these schemes are not suited for ``large'' MIMO
systems with many BS antennas serving many UTs.


In contrast, Time Division Duplexing (TDD) systems can exploit
channel reciprocity for estimating the downlink channels from uplink
training. In this case, the system multiplexing gain is still upper
bounded by $M^\star(1 - M^\star/T)$, but training in the same
coherence block is possible (hence, no extra degradation due to
prediction) and the training dimension is determined by the number
of total UT antennas, while the number of BS antennas can be made as
large as desired. By using $M \gg K$ antennas at the BS with TDD, as
proposed in \cite{Marzetta-TWC10} (see also the more refined
performance analysis and system optimization in
\cite{Huh11,debbah2012}), is very attractive for TDD systems both in
terms of achieved throughput and in terms of simplified downlink
scheduling and signal processing at the BS. Systems where the number
of BS antennas are much larger than the number of served UTs are
generally referred to as ``massive'' MIMO. A recent practical
testbed implementation of a 64 antenna massive MIMO system,
achieving transmitter clock stability and self-calibration in order
to effectively exploit TDD reciprocity, has been demonstrated in
\cite{argos}.


In this paper we consider a Joint Spatial Division and Multiplexing
(JSDM) approach to potentially achieve massive MIMO-like throughput
gains and simplified system operations also for FDD systems, which
still represent the far majority of currently deployed cellular
networks. We observe that, for a typical cellular
configuration, the channel from the $M$ BS antennas to any UT antenna is a {\em correlated} random vector with
covariance matrix that depends on the scattering geometry. Assuming a macro-cellular tower-mounted  BS with no significant local
scattering, the propagation between the BS antennas and any given UT antenna is characterized by the
local scattering around the UT, resulting in the well-known one-ring model \cite{Shi00}.
The main idea of JSDM consists of partitioning the user population into groups with {\em approximately}
the same channel covariance eigenspace, and split the downlink beamforming
into two stages: a pre-beamforming matrix that depends only on the channel
covariances, and a MU-MIMO precoding matrix  for the ``effective'' channel,
inclusive of pre-beamforming.  The pre-beamforming matrix is chosen in order to minimize the inter-group interference {\em for any instantaneous channel
realization}, by exploiting the linear independence of the {\em dominant eigenmodes}
of the channel covariance matrices of the different groups. Pre-beamforming can be considered as a
generalization of {\em sectorization}, widely used in current cellular technology.

The MU-MIMO precoding stage requires estimation and feedback of the instantaneous (effective) channel realization.
As we shall see, this may have significantly reduced dimension with respect to the original physical channel.
Therefore, both downlink training and uplink feedback overhead is greatly reduced, making this scheme attractive for FDD systems.
Notice that  the pre-beamforming stage requires only the channel covariance information,
which can be tracked with small protocol overhead.\footnote{In practice,
the channel covariance changes over time at a much slower time scale
with respect to the system slot rate, therefore we assume that this
is locally stationary and can be estimated and  tracked using some
standard subspace tracking technique \cite{vallet2012improved},
\cite{marzetta2011random}, \cite{hochwald2001adapting},
\cite{mestre2008improved}. See also the remark at the end of Section \ref{sec:jsdm}.}

We show that, under some conditions on the eigenvectors of the
channel covariance matrices, JSDM incurs no loss of optimality with
respect to the full CSIT case. When these conditions cannot be met,
we examine the design of the pre-beamforming matrix and the
performance of regularized zero forcing (linear) MU-MIMO precoding
for the resulting effective channel. Then, we specialize our system
design in the case of Uniform Linear Arrays (ULAs) and use Szego's
asymptotic results on Toeplitz matrices \cite{grenander1984toeplitz}
to show that the optimality conditions can be met by ULAs when $M$
is large, as long as the user groups have non-overlapping supports
of their Angle of Arrival (AoA) distributions. Using the Toeplitz
eigen-subspace approximation result of \cite{grenander1984toeplitz},
we argue that the pre-beamforming matrix for large ULAs can be
obtained by selecting blocks of columns of a unitary {\em Discrete
Fourier Transform} (DFT) matrix. DFT pre-beamforming achieves very
good performance and effective channel dimensionality reduction and
requires only a coarse knowledge of the support of the AoA
distribution for each user group, without requiring an accurate
estimation of the actual channel covariance matrix. Interestingly,
related eigen-structure properties of the covariance matrices were
independently derived in \cite{yin2012coordinated} for the purpose
of eliminating the {\em pilot contamination} effect which limits the
performance of TDD massive MIMO with the maximal-ratio single-user
beamforming advocated in \cite{Marzetta-TWC10}. Finally, we extend
our approach to the case of 2-dimensional ULAs (rectangular antenna
arrays) and three-dimensional (3D) beamforming, where we create
fixed beams also in the elevation angle direction, in addition to
the azimuth angle (planar) direction. The resulting beamforming
matrix takes on the appealing form of a Kronecker product. In this
way, we can serve simultaneously angular-separated groups of users
in different annular regions in a sector, at different distances
from the BS. We demonstrate the performance of such a system in a
realistic layout assuming a rectangular antenna array mounted on the
face of a tall building.

This paper focuses not only on the concept of JSDM, which is not
entirely new, but specifically on its performance analysis and system
design guidelines, i.e., how to choose the parameters of JSDM
for a given set of user groups that we wish to serve simultaneously, on the same
time-frequency slot. Since we focus on the large system regime,
we can leverage asymptotic random matrix theory results and in particular a
recently developed analytical tool referred to as ``deterministic
equivalent approximation'' (see \cite{couillet2011random} and
references therein), which is able to handle the rather complicated
class of structured random matrices arising in the JSDM context.
Thanks to this analytical tool, all numerical results presented here
are obtained in a semi-analytic way, by solving iteratively a provably
convergent system of fixed-point equations, without the need of
heavy Monte Carlo simulation.
For completeness, we provide the equations for the analysis of the basic JSDM schemes
without including channel estimation errors in Section \ref{sec:PERF}, and in Appendix \ref{sec:determ-equiv-nonideal-csi} the
corresponding general case including downlink estimation and noisy CSIT.

\emph{Notation :} We use boldface capital letters ($\Xm$) for matrices, boldface small letters for
vectors ($\xv$), and small letters ($x$) for scalars. $\Xm^\transp$ and $\Xm^\herm$  denote the transpose and the Hermitian
transpose of $\Xm$, $||\xv||$ denotes the vector 2-norm of $\xv$, $\trace(\Xm)$
and $|\Xm|$ denote the trace and the determinant of the square matrix $\Xm$.
The identity matrix is denoted by $\Id$ (when the dimension is clear from the context)
or by $\Id_n$ (when pointing out its dimension $n \times n$ improves clarity of exposition).
$\Xm \otimes \Ym$ denotes the Kronecker product of two matrices $\Xm,\Ym$.
$\|\Xm\|_F^2 = \trace(\Xm^\herm\Xm)$ indicates the squared Frobenius norm of a matrix $\Xm$.
We also use ${\rm Span}(\Xm)$ to denote the linear subspace generated by columns
of $\Xm$ and ${\rm Span}^\perp(\Xm)$ for the orthogonal complement
of ${\rm Span}(\Xm)$. $\xv \sim \Cc\Nc(\muv, \Sigmam)$ indicates that $\xv$ is a complex circularly-symmetric Gaussian vector
with mean $\muv$ and covariance matrix $\Sigmam$.

\section{Channel Model}  \label{sec:channel-model}

We consider the downlink of a single-cell FDD system
with a BS with $M$ antennas serving $K$ UTs equipped with a single
antenna each.  For simplicity, we consider a narrowband (frequency-flat) channel model.
By using the Karhunen-Loeve representation, a generic  downlink channel
vector from the $M$ BS antennas to a UT can be expressed as
\begin{align} \label{eq:SM-2}
   \ph={\pU}{\boldsymbol \Lambda}^{\frac{1}{2}} \wv,
\end{align}
where $\wv \in \mathbb{C}^{r \times 1} \sim\mathcal{CN} (\p0, \pI)$,
${\boldsymbol \Lambda}$ is an $r\times r$ diagonal matrix whose
elements are the non-zero eigenvalues of $\pR$, and $\pU \in
\mathbb{C}^{M\times r}$ is the tall unitary matrix of the eigenvectors of $\pR$ corresponding to the non-zero eigenvalues.
We consider the one-ring model of Fig.~\ref{one-ring-model}, where a UT located at
azimuth angle $\theta$ and distance ${\sf s}$ is surrounded by a
ring of scatterers of radius  ${\sf r}$ such that the AS is $\Delta \approx \arctan({\sf r}/{\sf s})$.
Assuming a uniform distribution\footnote{The uniform distribution is assumed here only for analytical convenience. It is
easy to show that similar performances and asymptotic behaviors are achieved by any AoA distribution (measurable non-negative
function integrating to 1) with limited support in $[\theta - \Delta, \theta + \Delta]$.}
of the received power from planar waves impinging on the BS antennas, the correlation
between the channel coefficients of antennas $1 \leq m, p \leq M$ is given by (see \cite{Shi00} and references therein)
\begin{align} \label{eq:SM-4}
   [\pR]_{m,p} = \frac{1}{2\Delta}  \int_{-\Delta}^{\Delta} e^{ j \kv^\transp(\alpha + \theta) (\uv_m - \uv_p) } d\alpha,
\end{align}
where $\kv(\alpha) =  - \frac{2\pi}{\lambda} (\cos(\alpha),
\sin(\alpha))^\transp$ is the wave vector for a planar wave
impinging with AoA $\alpha$, $\lambda$ is the carrier wavelength,
and $\uv_m, \uv_p \in \RR^2$ are the vectors indicating the position
of BS antennas $m,p$ in the two-dimensional coordinate system (see Fig.~\ref{one-ring-model}).

\begin{figure}[ht]
\centerline{\includegraphics[width=8cm]{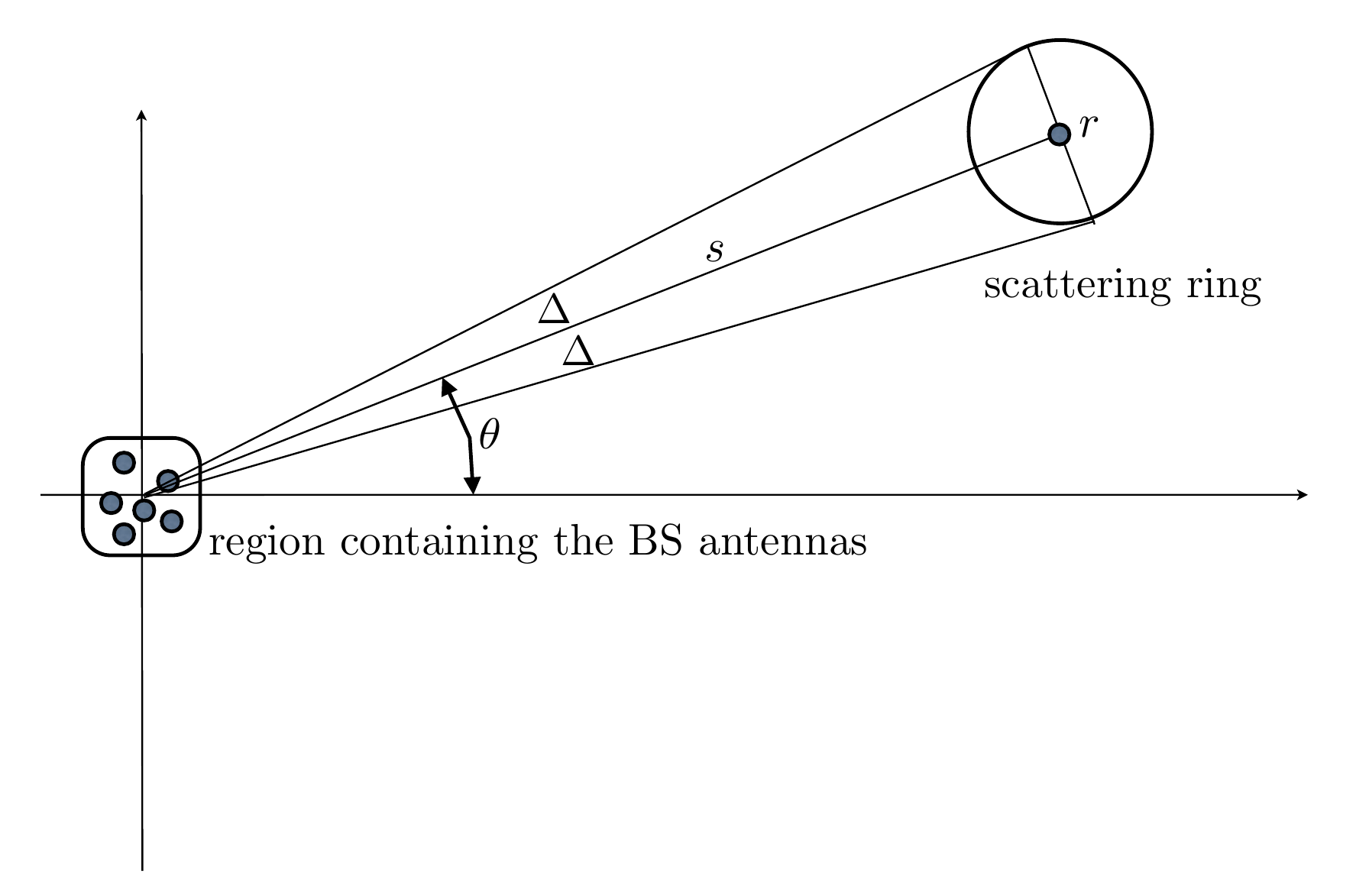}} \caption{A
UT at AoA $\theta$ with a scattering ring of radius ${\sf r}$
generating a two-sided AS $\Delta$ with respect to the BS at the
origin.} \label{one-ring-model}
\end{figure}

Let $\underline{\pH}$ denote the $M\times K$ system channel matrix
given by stacking the $K$ users channel vectors by columns. The
signal vector received by the UTs is given by
\begin{align} \label{eq:SM-3}
   \py=\underline{\pH}^\herm\Vm\pd +\pz =\underline{\pH}^\herm\px +\pz
\end{align}
where $\Vm$ is the $M\times S$ precoding matrix with $S$ the rank of
the input covariance $\boldsymbol{\Sigma}=\mathbb{E}[\Vm
\pd\pd^\herm\Vm^\herm]$ (i,e., the number of independent
data streams sent to the users), $\pd$ is the $S$-dimensional
transmitted data symbol vector, and $\pz  \sim\mathcal{CN} (\p0,
\pI)$ denotes the Gaussian noise at the UT receivers. The transmit
signal vector is given by $\px=\Vm \pd$.

\section{Joint Spatial Division and Multiplexing} \label{sec:jsdm}

JSDM exploits the fact that, after appropriate partitioning of the
UTs such that users in the same group are nearly co-located and
different groups are sufficiently well separated in the AoA domain,
the structure of the channel covariance matrices can be leveraged in
order to reduce the dimensionality of the effective channels and
therefore achieve large multiplexing gains with reduced dimension
channel training and CSIT feedback.

Suppose that $K$ UTs are selected to form $G$ groups based on the similarity of their channel covariance matrices.
We let $K_g$ denote the number of UTs in group $g$, such that $K = \sum_{g=1}^G K_g$, and define
the index $g_k = \sum_{g'=1}^{g-1} K_{g'} + k$, for $k = 1,\ldots,K_g$, to denote UT $k$ in group $g$.
Similarly, we let $S_g$ denote the number of independent data streams sent to users in group $g$, such that
$S = \sum_{g=1}^G S_g$.
We assume for simplicity that all UTs in the same group $g$ have identical covariance matrix
$\Rm_g = \Um_g \Lambda_g \Um_g^\herm$, with rank $r_g$ and $r^\star_g \leq r_g$
dominant  eigenvalues. In practice, this condition is not verified exactly, but we can select groups such that
this condition is closely approximated. Also, the notion of ``dominant eigenvalues'' is intentionally left fuzzy,
since $r^\star_g$ is a design parameter that depends on how much signal power outside the subspace
spanned by the corresponding eigenvectors can be tolerated.
For future reference, we denote by $\Um_g^\star$ the $M \times r^\star_g$ matrix collecting
the dominant eigenvectors, and let $\Um_g = [\Um_g^\star , \Um_g']$, with $\Um_g'$ of dimension
$M \times (r_g - r^\star_g)$, containing the eigenvectors corresponding to the weakest eigenvalues.
Notice that, by construction, we have that $0 \leq S_g \leq \min\{K_g, r^\star_g\}$, since we cannot deliver more
independent symbol streams than the multiplexing gain $\min\{K_g, r^\star_g\}$ of each group $g$.

The channel vector of user $g_k$ is given by $\ph_{g_k} = {\pU}_{g}\Lambdam_{g}^{\frac{1}{2}} \wv_{g_k}$. We let $\pH_g =
\big[{\ph}_{g_1}, \cdots,{\ph}_{g_{K_g}}\big]$ and $\underline{\pH}=\big[{\pH}_{1}, \cdots,{\pH}_{G}\big]$ denote the
group $g$ channel matrix and the overall system channel matrix, respectively.
As anticipated in Section \ref{sec:intro}, JSDM is based on two-stage precoding. Namely, we let
$\Vm = \pB \pP$, where $\pB \in\mathbb{C}^{M \times b}$ is a {\em pre-beamforming} matrix,
$\Pm \in \CC^{b \times S}$ is a MU-MIMO precoding matrix, and where $b \geq S$ is an integer design parameter,
to be optimized. The pre-beamforming matrix $\Bm$ is a function of the channels second-order statistics, i.e., it depends on the set $\{ \pU_{g}, \Lambdam_{g}\}$,
or on some directional information extracted from the channel covariance matrices (AoA and AS of the different groups).
In any case, $\Bm$ is independent of the instantaneous realization of the channel matrix $\underline{\pH}$.
The MU-MIMO precoding matrix $\pP$ is allowed to depend on the instantaneous realization of the reduced dimensional
\emph{effective channel} $\underline{\textsf{\pH}} \triangleq  \pB^\herm \underline{\pH}$.
We let $b = \sum_{g=1}^G b_g$ such that $b_g \geq S_g$,
and let $\pB_g$ be the $M\times b_g$ pre-beamforming matrix of group $g$. The received signal (\ref{eq:SM-3}) can be rewritten as
\begin{align} \label{eq:SM-4b}
\py=\underline{\textsf{\pH}}^\herm\pP\pd +\pz
\end{align}
where
$$\underline{\textsf{\pH}}^\herm =\left [\begin{matrix} \pH_1^\herm\pB_1 & \pH_1^\herm\pB_2 & \cdots & \pH_1^\herm\pB_G \\ \pH_2^\herm\pB_1 & \pH_2^\herm\pB_2 & \cdots & \pH_2^\herm\pB_G \\  \vdots & \vdots & \ddots & \vdots \\\pH_G^\herm\pB_1 & \pH_G^\herm\pB_2 & \cdots & \pH_G^\herm\pB_G \end{matrix} \right ], $$
and where $\pH_{g}^\herm \pB_{g'}$ is the $K_{g} \times b_{g'}$ effective channel matrix connecting the users of group $g$ with the effective channel
inputs of group $g'$.

If the estimation and feedback of the effective channel $\underline{\textsf{\pH}}$ can be afforded,
the precoding matrix $\Pm$ is determined as a function of the whole $\underline{\textsf{\pH}}$.
We refer to this approach as {\em Joint Group Processing} (JGP). However, this
may still be too costly in terms of transmission resource. Hence, a lower complexity and generally  more attractive approach
consists of estimating and feeding back only the $G$ diagonal  blocks $\textsf{\pH}_g = \Bm_g^\herm \Hm_g$, of dimension $b_g \times K_g$,
and treating each group separately. We refer to this approach as {\em Per-Group Processing} (PGP). In this case, the precoding matrix takes on the block-diagonal form
$\pP=\mathrm{diag}(\pP_1,\cdots,\pP_G)$, where $\pP_g \in \mathbb{C}^{b_g \times S_g}$, resulting in the vector broadcast plus
interference Gaussian channel
\begin{align} \label{eq:SM-5-approx}
   \py_g = \textsf{\pH}_g^\herm  \pP_g \pd_g +  \sum_{g' \neq g}  {\pH}_g{}^\herm\pB_{g'} \pP_{g'}\pd_{g'}   +  \pz_g,  \;\;\; \mbox{for} \;\;  g=1,\ldots, G.
\end{align}
With PGP, it is interesting to choose the groups and design the pre-beamforming matrix  such that, with high probability,
\begin{align} \label{eq:SM-10}
   {\pH}_g{}^\herm\pB_{g'}\approx \p0, \ \text{for all} \ g'\neq g.
\end{align}
Exact Block Diagonalization (BD) is possible if ${\rm Span}(\Um_g) \not\subseteq {\rm Span}(\{\Um_{g'} : g' \neq g\})$
for all $g = 1,\ldots, G$. In particular, multiplexing gain $S_g$ (i.e., the number of
interference-free data streams) can be achieved for group $g$ if and only if
\begin{equation}
{\rm dim} \left ( {\rm Span}(\Um_g) \cap {\rm Span}^\perp(\{\Um_{g'} : g' \neq g\}) \right ) \geq S_g.
\end{equation}
Approximate BD can be achieved by selecting $r^\star_g$ dominant eigenmodes for each group $g$, such that
${\rm Span}(\Um^\star_g) \not\subseteq {\rm Span}(\{\Um^\star_{g'} : g' \neq g\})$ for all $g = 1,\ldots, G$.
In this case, in order to deliver $S_g$ streams to group $g$ we require
\begin{equation}
{\rm dim} \left ( {\rm Span}(\Um^\star_g) \cap {\rm Span}^\perp(\{\Um^\star_{g'} : g' \neq g\}) \right ) \geq S_g.
\end{equation}
However, these streams will be affected by some residual interference due to the weak eigenmodes  not included in
the matrices $\{\Um^\star_g: g = 1,\ldots, G\}$.

\begin{rem}
Notice that the PGP pre-beamforming creates \emph{virtual sectors}, i.e., a generalization of spatial sectorization commonly
used in current cellular technology. Each group corresponds to a virtual sector, and it is independently precoded
under a total sum power constraint, possibly incurring some residual {\em inter-group interference} in the case of
approximate  BD.
\hfill $\lozenge$
\end{rem}

\begin{rem}
It is reasonable to assume that the channel covariance matrix $\Rm_g$ for
each user group changes slowly with respect to the
coherence time of the instantaneous channel matrix $\Hm_g$.
The dominant eigenmodes $\Um_g^\star$ can be tracked
for each UT using a suitable subspace estimation and tracking algorithm
\cite{eriksson1994line}, by exploiting the downlink training phase,
and they can be fed back to the BS at a low rate.
Furthermore, for particularly designed BS antenna configurations,
these estimates can be refined at the BS by exploiting the uplink,
even though in an FDD system this takes place at a different carrier frequency
(see for example \cite{Hoc01}).
The estimation and tracking of the (slowly time-varying) channel statistics
is a topic of great interest in this context, but it is out of the scope of this paper.
Here, we assume that the channel covariance matrix for each user is known.
\hfill $\lozenge$
\end{rem}


\section{JSDM with Eigen-Beamforming}
\label{sec:MU}

\subsection{Achieving capacity with reduced CSIT}

Let $r = \sum_{g=1}^G r_g$ and suppose that the channel covariances of the $G$ groups are such that
$\underline{\pU}=[\pU_1,\cdots,\pU_G]$ is $M \times r$ \emph{tall unitary}
(i.e., $r \leq M$ and $\underline{\pU}^\herm \underline{\pU} = \Id_r$).
In order to obtain {\em exact} BD it is sufficient to let $b_g = r_g$ and $\Bm_g = \pU_g$. This choice for the pre-beamforming matrix
is referred to in the following as {\em eigen-beamforming}. In this case, the decoupled MU-MIMO channel
(\ref{eq:SM-5-approx}) takes on the form
\begin{align} \label{eq:SM-5}
   \py_g & = \textsf{\pH}_g{}^\herm \pP_g\pd_g +\pz_g = \Wm_g^\herm \Lambda_g^{1/2} \pP_g\pd_g +\pz_g, \;\;\; \;\;\; \mbox{for} \;\;  g=1,\ldots, G,
\end{align}
where $\Wm_g$ is a $r_g \times K_g$ i.i.d. matrix with elements $\sim
\Cc\Nc(0,1)$. In this case we have:

\begin{thm} \label{simple-opt}
For $\underline{\pU}$ tall unitary, JSDM with PGP achieves the same sum capacity of the corresponding MU-MIMO
downlink channel (\ref{eq:SM-3}) with full CSIT.
\end{thm}

\begin{proof}
Let $\Cc^{\rm sum}(\underline{\Hm}; P)$ denote the sum capacity of
(\ref{eq:SM-3}) with sum power constraint $P$ and fixed channel
matrix $\underline{\Hm}$, perfectly known to transmitter and
receivers. By the MAC-BC duality \cite{vishwanath2003duality},  we
have
\begin{align} \label{sum-full-csi}
\Cc^{\rm sum}(\underline{\Hm}; P) & = \max_{\Sm_g \succeq 0: \sum_g \trace(\Sm_g) \leq P}
 \;\; \log \left | \Id_M + \sum_{g=1}^G \Um_g \Lambdam_g^{1/2} \Wm_g \Sm_g \Wm_g^\herm \Lambdam_g^{1/2}\Um_g^\herm \right | \nonumber \\
\end{align}
where $\Sm_g$ denotes the diagonal $K_g \times K_g$ input covariance matrix for group $g$ in the dual MAC channel.
For any fixed set $\{\Sm_g\}$ of feasible input covariance matrices,
define for notation simplicity $\Am_g = \Lambdam_g^{1/2} \Wm_g \Sm_g \Wm_g^\herm \Lambdam_g^{1/2}$.
Notice that $\Am_g$ has dimension $r_g \times r_g$ and is invertible
with probability 1 over the random channel realization.
The theorem is proved by showing the
the determinant identity
\begin{equation} \label{determ-identity}
\left | \Id_M + \sum_{g=1}^G \Um_g \Am_g \Um_g^\herm \right | = \prod_{g=1}^G \left | \Id_M +  \Um_g \Am_g \Um_g^\herm \right |.
\end{equation}
This can be proved by induction, noticing the following step: for any $1 \leq g' \leq G$,
\begin{eqnarray}
\left | \Id_M + \sum_{g=g'}^G \Um_g \Am_g \Um_g^\herm \right | &  = & \left | \Id_M + \Um_{g'} \Am_{g'} \Um_{g'}^\herm \right | \left | \Id_M + (\Id_M + \Um_{g'} \Am_{g'} \Um_{g'}^\herm)^{-1} \sum_{g=g'+1}^G \Um_g \Am_g \Um_g^\herm \right | \nonumber \\
& = & \left | \Id_M + \Um_{g'} \Am_{g'} \Um_{g'}^\herm \right |
\left | \Id_M + (\Id_M - \Um_{g'} ( \Am_{g'}^{-1} + \Id_{r_g})^{-1} \Um_{g'}^\herm) \sum_{g=g'+1}^G \Um_g \Am_g \Um_g^\herm \right |\label{bunga} \\
& = & \left | \Id_M + \Um_{g'} \Am_{g'} \Um_{g'}^\herm \right | \left | \Id_M + \sum_{g=g'+1}^G \Um_g \Am_g \Um_g^\herm \right |, \label{bunga1}
\end{eqnarray}
where (\ref{bunga}) follows form the matrix inversion lemma and (\ref{bunga1}) follows from the the fact that, by assumption,
$\Um_{g'}^\herm \Um_g = \zerov$ for all $g' \neq g$. Using (\ref{determ-identity}) in (\ref{sum-full-csi}) we
obtain
\begin{align} \label{sum-full-csi1}
\Cc^{\rm sum}(\underline{\Hm}; P) & = \max_{\Sm_g \succeq 0: \sum_g \trace(\Sm_g) \leq P} \;\; \sum_{g=1}^G \log \left | \Id_{r_g} + \Lambdam_g^{1/2} \Wm_g \Sm_g \Wm_g^\herm \Lambdam_g^{1/2} \right |,
\end{align}
which is immediately recognized to be the capacity of the dual MAC (with sum power constraint)
for the set of decoupled MU-MIMO downlink channels (\ref{eq:SM-5}).
\end{proof}

\begin{rem}
In a similar manner it is possible to show that under the
orthogonality condition of Theorem \ref{simple-opt}, JSDM achieves
the whole capacity region \cite{weingarten2006capacity}, and not
only the sum capacity. In order to see this, for any user subset
$\Kc \subseteq \{1, \ldots, K\}$ define $\Hm_g(\Kc)$ as the sub
matrix of $\Hm_g$ obtained by selecting the columns $g_k \in \Kc$,
and let $\Sm_g(\Kc)$ denote the submatrix of $\Sm_g$ obtained by
retaining the rows and columns corresponding to users $g_k \in \Kc$.
Then, the capacity region of the dual MAC of (\ref{eq:SM-3}) subject
to the sum power constraint can be written as
\begin{equation} \label{dual-MAC-region}
\Cc(\underline{\Hm}; P) = \bigcup_{\substack{\Sm_g \succeq 0 : \\
\sum_{g=1}^G \mathrm{Tr}(\Sm_g) \le P}} \left \{ \rv \in \RR_K^+ : \sum_{g_k \in \Kc} r_{g_k} \leq
\log \left | \Id_M + \sum_{g=1}^G \Hm_g(\Kc) \Sm_g(\Kc) \Hm^\herm_g(\Kc) \right | , \; \forall \; \Kc \subseteq \{1,\ldots, K\} \right \}.
\end{equation}
The determinant identity (\ref{determ-identity}) can be applied to the partial sum-rate
bounds for each user subset $\Kc$, noticing that the tall unitary condition of the singular vectors
is retained by the new system matrix $\underline{\Hm}(\Kc) = [\Hm_1(\Kc), \ldots, \Hm_G(\Kc)]$.
\hfill $\lozenge$
\end{rem}

\begin{rem} \label{scheduling-remark}
Theorem \ref{simple-opt} has an important practical implication: in a situation where a large number of UTs, each of which has its own
AoA and AS, must  be served by the downlink, a good scheduling strategy consists of the following.
First, partition the users into groups with (approximately) identical eigenspaces.
Then, partition the collection of groups into disjoint and mutually exclusive sets, such that the groups in each set satisfy
the tall unitary condition of Theorem \ref{simple-opt}, and such that the number of sets is minimal,
over all possible partitions. Finally, schedule the groups in each set to be served simultaneously,
on the same time-frequency slot, using JSDM, and use time-frequency sharing across the groups.
Notice that this does not mean that, in general, JSDM is optimal. In fact, in order to meet the tall unitary condition
we may be obliged to reduce the number $G$  of simultaneously served groups in each set.
As already noticed for the problem of clustering users into groups,  also the problem of finding optimal partitions of the user groups
under JSDM with PGP is far from trivial, and goes beyond the scope of this paper.
\hfill $\lozenge$
\end{rem}

When achieving the tall unitary condition is too restrictive in
terms of multiplexing gain, the pre-beamforming matrix $\Bm$ can be
chosen as a function of the whole $\underline{\Um}$ in order to
achieve exact or approximated BD. This approach is presented in the
next section.

\subsection{Block diagonalization}
\label{sec:BD}

Recall that $\Bm = \left[\Bm_1,\ldots,\Bm_G \right]$ is an $M \times
b$ matrix consisting of $G$ blocks of dimension $M \times b_g$, each
corresponding to a particular group $g$. For given target numbers of
streams per group $\{S_g\}$ and dimensions $\{b_g\}$ satisfying $S_g
\leq b_g \leq r_g$, our goal is to design the blocks $\Bm_g$  such
that BD is achieved, i.e., $\Um_{g'}^\herm \Bm_g = \zerov$ for all
$g' \neq g$ and ${\rm rank}(\Um_g^\herm \Bm_g) \geq S_g$. A
necessary condition for exact zero-forcing of the off-diagonal
blocks is ${\rm Span}(\Bm_g) \subseteq  {\rm Span}^\perp(\{\Um_{g'}
: g' \neq g\})$. When ${\rm Span}^\perp(\{\Um_{g'} : g' \neq g\})$
has dimension smaller than $S_g$, the rank condition on the diagonal
blocks cannot be satisfied. In this case, $S_g$ should be reduced
or, as an alternative, approximated BD based on selecting $r^\star_g
< r_g$ dominant eigenmodes for each group $g$ can be implemented.
This consists of replacing $\Um_g$ with $\Um_g^\star$ in the above
conditions. When ${\rm Span}(\{\Um_{g'} : g' \neq g\})$ has
dimension $M$, then exact BD cannot be achieved even for $S_g = 1$,
and therefore approximated BD should be considered in any case.
Without loss of generality, we formulate the design of $\{\Bm_g\}$
for approximated BD with some feasible choice of the parameters
$\{r^\star_g\}$, $\{b_g\}$ and $\{S_g\}$. It should be noticed that
these are design parameters that should be optimized for a given
system configuration, in order to maximize the overall spectral
efficiency. This optimization is far from trivial. For the time
being, we consider an arbitrary feasible choice and postpone the
discussion on the tradeoff that governs the design of these
parameters in Sections \ref{subsec:res-jsdm} (see Remark
\ref{choice-of-rstar-remark}) and \ref{sec:prime-tradeoff} (see
Remark \ref{choice-of-bprime-remark}).

Following the approach of \cite{spencer2004zero}, we  define
\begin{equation}
\Xim_g =
\left[\Um_1^\star,\ldots,\Um_{g-1}^\star,\Um_{g+1}^\star,\ldots,\Um_G^\star\right],
\end{equation}
of dimensions $M \times \sum_{g' \neq g} r_{g'}^\star$ and rank $\sum_{g' \neq g} r_{g'}^\star$,
and let $[\Em_g^{(1)}, \Em_g^{(0)}]$ denote a system of left eigenvectors of $\Xim_g$
(e.g., obtained by Singular Value Decomposition (SVD)), such that
$\Em_g^{(0)}$ is $M \times \left (M - \sum_{g' \neq g} r_{g'}^\star \right )$ and forms a unitary basis for the orthogonal complement of
${\rm Span}(\Xim_g)$, i.e., such that ${\rm Span}(\Em_g^{(0)}) = {\rm Span}^\perp(\{\Um^\star_{g'} : g' \neq g\})$.

We obtain $\Bm_g$ by concatenating the projection onto ${\rm
Span}(\Em_g^{(0)})$ with eigen-beamforming along the dominant
eigenmodes of the covariance matrix of the resulting projected
channels of group $g$, i.e., of the columns of $(\Em_g^{(0)})^\herm
\Hm_g$. Recalling the Karhunen-Loeve decomposition (\ref{eq:SM-2}),
we have that the covariance matrix of $\widehat{\hv}_{g_k} =
(\Em_g^{(0)})^\herm \Um_g \Lambdam_g^{1/2} \wv_{g_k}$ is given by
\begin{equation}  \label{svdsvd}
\widehat{\Rm}_{g} = (\Em_g^{(0)})^\herm \Um_g \Lambdam_g \Um_g^\herm  \Em_g^{(0)} = \Gm_g \Phim_g \Gm_g^\herm,
\end{equation}
where the expression on the right of (\ref{svdsvd}) is the SVD of $\widehat{\Rm}_{g}$.
Letting $\Gm_g = [\Gm_g^{(1)}, \Gm^{(0)}_g]$ where $\Gm^{(1)}_g$ contains the
dominant $b_g$ eigenmodes of $\widehat{\Rm}_{g}$, we eventually obtain
\begin{equation}
\label{eq:design-Bg}
\Bm_g = \Em_g^{(0)} \Gm_g^{(1)}.
\end{equation}
The pre-beamforming matrix $\Bm_g$ can be interpreted as being orthogonal to the dominant
$r^{*}_{g'}$ eigenmodes  of groups $g' \neq g$, and matched to the $b_g$ dominant eigenmodes
of the covariance matrix of the projected channels $(\Em_g^{(0)})^\herm \Hm_g$ of group $g$.
By construction, we have that $b_g$ is less or equal to the rank of $\widehat{\Rm}_{g}$, given by
$\min\left \{ r_g, M - \sum_{g' \neq g} r_{g'}^\star \right \}$. In particular,
if $r = \sum_g r_g \leq M$, we can choose $b_g = r^\star_g = r_g$ and obtain exact BD.

\section{Performance analysis with linear precoding}
\label{sec:PERF}

In this section we provide expressions for the performance analysis
of JSDM with JGP and PGP and linear precoding, using the techniques
of deterministic equivalents \cite{debbah2012}. For simplicity of
exposition, we consider a symmetric scenario with the same number
$K_g = K'$ of users per group, the same number $S_g = S'$ of streams
per group, and the same dimension $b_g = b'$ of the pre-beamforming
matrix per group. However, the analysis can be immediately extended
to the general case considered before. This technique can be applied
as long as the users to be served in each group are selected
independently of their instantaneous channel realization. Hence, we
assume that for each group a subset of $S'$ out of the possible $K'$
users is pre-selected and scheduled for transmission over the
current downlink time-frequency slot. This simplified scheduling
requires only the instantaneous CSIT feedback from the pre-scheduled
users~\footnote{Unlike channel-based opportunistic user selection,
\cite{viswanath2002opportunistic,yoo2006optimality,al2009much,sharif2005capacity},
that requires to collect CSIT from many users and then select a
subset of users with quasi-orthogonal channel vectors.} and it is in
line with the massive MIMO concept, where hardware augmentation at
the BS allows significant simplification in the system operations.

Under these assumptions, the transformed channel matrix
$\underline{\textsf{\pH}}$ has dimension $b \times S$, with blocks
${\textsf{\pH}}_g$ of dimension $b' \times S'$. Also for the sake of
simplicity and in line with massive MIMO system simplification (see
for example \cite{Marzetta-TWC10,Huh11}) we allocate to all users
the same fraction of the total transmit power $P$, such that the
data vector covariance matrix is given by $\EE[\dv \dv^\herm] =
\frac{P}{S} \Id_S$. In the following, we present the deterministic
equivalent fixed-point equations for determining the {\em
Signal-to-Interference plus Noise Ratio} (SINR) at the UTs receivers
for the case of JSDM with JGP and PGP with linear {\em regularized
zero forcing precoding}. Along the same lines,
Appendix \ref{sec:determ-equiv-nonideal-csi} presents  the case of
regularized and non-regularized linear zero forcing precoding for PGP in the case of
noisy CSIT obtained from downlink training (see Section \ref{sec:nonperfect-csi}).
It is well-known that a discrete-time complex additive noise plus
interference channel with SINR  equal to $\gamma$ has capacity at
least as large as $\log(1 + \gamma)$ bit/symbol
\cite{merhav1994information}. Hence, in order to obtain an
asymptotically convergent approximation of the  achievable spectral
efficiency (in bit/symbol) per served user,  we compute $\gamma$ via
the deterministic equivalent  method, and plug the result into the
$\log(1 + \gamma)$ rate formula.

\subsection{JSDM with joint group processing}
\label{sec:JGP}

For fixed pre-beamforming matrix $\Bm$ and JGP, the regularized zero forcing precoding matrix is given by
\begin{equation}
\label{rzfbf-jsdm} \Pm_{{\rm rzf}} = \zeta \Km \underline{\textsf{\pH}},
\end{equation}
where $\Km =  \left [ \underline{\textsf{\pH}} \underline{\textsf{\pH}}^\herm
 + b \alpha \Id_{b} \right ]^{-1}$, $\alpha$ is a regularization factor, and $\zeta$ is a normalization factor chosen
 to satisfy the power constraint and is given by
\begin{equation}
\zeta^2 = \frac{S}{\trace \left ( \Pm_{\rm rzf}^\herm \Bm^\herm
\Bm \Pm_{\rm rzf} \right )}.
\end{equation}
The covariance matrix of the transformed channel of group $g$ is given by
\begin{equation}
\label{ch-cov-jsdm} \tilde{\Rm}_g = \left[
\begin{matrix}
\Bm_1^\herm \Rm_g \Bm_1 & \Bm_1^\herm \Rm_g \Bm_2 & \cdots & \Bm_1^\herm \Rm_g \Bm_G \\
\Bm_2^\herm \Rm_g \Bm_1 & \Bm_2^\herm \Rm_g \Bm_2 & \cdots & \Bm_2^\herm \Rm_g \Bm_G \\
\vdots & \vdots & \ddots  & \vdots \\
\Bm_G^\herm \Rm_g \Bm_1 & \Bm_G^\herm \Rm_g \Bm_2 & \cdots &
\Bm_G^\herm \Rm_g \Bm_G
\end{matrix}
\right ].
\end{equation}
The SINR for user $g_k$ is given by
\begin{equation}
\label{sinr-jsdm-rzf-1} {\rm \gamma}_{g_k, {\rm jgp,rzf}} =
\frac{\frac{P}{S}\zeta^2 |\hv_{g_k}^\herm \Bm \Km \Bm^\herm
\hv_{g_k}|^2}{ \frac{P}{S} \sum_{j \neq g_k} \zeta^2 |\hv_{g_k}^\herm \Bm \Km
\Bm^\herm \hv_j |^2 + 1}
\end{equation}
where the subscript ``jgp'' stands for {\em joint group processing}.

Following the approach of \cite{debbah2012}, assuming that as $M \rightarrow \infty$ the other system dimensions
$r, S$ and $b$ also go to infinity linearly with $M$, we have
\begin{equation}
\gamma_{g_k,{\rm jgp,rzf}} - \gamma_{g_k,{\rm jgp,rzf}}^o
\stackrel{M \rightarrow \infty}{\longrightarrow} 0 \;\;\; \mbox{with probability 1},
\end{equation}
where, for all users $g_k$, $\gamma_{g_k,{\rm jgp,rzf}}^o$ is a deterministic quantity that can be computed for any finite $M$ as
\begin{equation}
\gamma_{g_k,{\rm jgp,rzf}}^o = \frac{\frac{P}{S} \zeta^2
(m_g^o)^2}{\zeta^2 \Upsilon_g^o + (1 + m_g^o)^2},
\end{equation}
where $\zeta^2 = \frac{P}{\Gamma^o}$ and the quantities
$m_g^o$, $\Upsilon_g^o$ and $\Gamma^o$ are obtained by solving the system of fixed-point equations
\begin{eqnarray}
\label{fixed-pt-1-rzfbf-jsdm} m_g^o &=& \frac{1}{b} \trace \left ( \tilde{\Rm}_g \Tm \right ) \\
\label{fixed-pt-2-rzfbf-jsdm} \Tm &=& \left( \frac{S'}{b} \sum_{g=1}^{G} \frac{\tilde{\Rm}_{g}}{1 + m_{g}^o} +
\alpha \Id_{b}\right)^{-1}\\
\Gamma^o &=& \frac{1}{b} \frac{P}{G} \sum_{g = 1}^{G} \frac{
n_{g}}{(1 +
m_{g}^o)^2}\\
\Upsilon_g^o &=& \frac{1}{b} \frac{P}{G}
\left[\sum_{g'=1, g' \neq g}^{G} \frac{ n_{g',g}}{(1 +
m_{g'}^o)^2} + \frac{S' - 1}{S'} \frac{ n_{g,g}}{(1 + m_g^o)^2} \right], \nonumber \\
& &
\end{eqnarray}
with $\nv =
[n_1,n_2,\ldots,n_{G}]^\transp$ and $\nv_g =
[n_{1,g},n_{2,g},\ldots,n_{G,g}]^\transp$ defined by
\begin{eqnarray}
\nv &=& (\Id_{G} - \Jm)^{-1} \vv\\
\nv_g &=& (\Id_{G} - \Jm)^{-1} \vv_g,
\end{eqnarray}
where $\Jm,\vv$ and $\vv_g$ are given as
\begin{eqnarray}
[\Jm]_{g,g'} &=& \frac{\frac{S'}{b} \trace\left ( \tilde{\Rm}_g
\Tm \tilde{\Rm}_{g'} \Tm \right )}{b (1 + m_{g'}^o)^2} \label{funny} \\
\vv &=& \frac{1}{b}\left[  \trace \left ( \tilde{\Rm}_1 \Tm \Bm^\herm \Bm \Tm \right ),\ldots,
\trace\left ( \tilde{\Rm}_{G} \Tm \Bm^\herm \Bm \Tm \right ) \right]^\transp \\
\vv_g &=& \frac{1}{b} \left[  \trace\left ( \tilde{\Rm}_1
\Tm \tilde{\Rm}_g \Tm\right ),\ldots, \trace\left (
\tilde{\Rm}_{G} \Tm \tilde{\Rm}_g \Tm \right )\right]^\transp
\end{eqnarray}

\subsection{JSDM with per-group processing}
\label{sec:PGP}

The channel covariance matrix for a user $g_k$ is given
by $\bar{\Rm}_{g} = \Bm_g^\herm \Rm_g \Bm_g$. Focusing only on the
users in group $g$, the regularized zero forcing precoding matrix is given by
\begin{equation}
\label{rzfbf-jsdm-sect} \Pm_{g,{\rm rzf}} = \bar{\zeta}_g
\bar{\Km}_g
{\textsf{\pH}}_g,
\end{equation}
where $\bar{\Km}_g = \left [ {\textsf{\pH}}_g{\textsf{\pH}}_g^\herm + b' \alpha \Id_{b'} \right ]^{-1}$, $\alpha$ is a regularization factor, and
$\bar{\zeta}_g$ is the power normalization factor given by
\begin{equation} \label{power-factor-pgp}
\bar{\zeta}_g^2 = \frac{S'}{\trace\left( \Pm_{g,{\rm rzf}}^\herm
\Bm_g^\herm \Bm_g \Pm_{g,{\rm rzf}}\right )}.
\end{equation}
When $\Bm_g$ is given by (\ref{eq:design-Bg}), then it is the product of two tall unitary matrices
so that $\Bm_g^\herm \Bm_g = \Id_{b'}$. However, we use (\ref{power-factor-pgp}) for the sake of generality.

The SINR of user $g_k$ given by
\begin{equation}
\label{sinr-jsdm-sect-rzf-1}
{\rm \gamma}_{g_k,{\rm pgp}} =
\frac{\frac{P}{S} \bar{\zeta}_g^2 |\hv_{g_k}^\herm \Bm_g \bar{\Km}_g
\Bm_g^\herm \hv_{g_k}|^2} {\frac{P}{S} \sum_{j \neq k}
\bar{\zeta}_g^2 |\hv_{g_k}^\herm \Bm_g \bar{\Km}_g \Bm_g^\herm
\hv_{g_j}|^2  + \frac{P}{S} \sum_{g' \neq g} \sum_j \bar{\zeta}_{g'}^2
|\hv_{g_k}^\herm \Bm_{g'} \bar{\Km}_{g'} \Bm_{g'}^\herm
\hv_{g'_j}|^2 + 1}
\end{equation}
where the subscript ``pgp'' stands for {\em per-group processing}.

Proceeding similarly as before and applying the method developed in \cite{debbah2012},
and assuming that as $M \rightarrow \infty$ the other system dimensions
$r, S$ and $b$ also go to infinity linearly with $M$, we have
\begin{equation}
\gamma_{g_k,{\rm pgp,rzf}} - \gamma_{g_k,{\rm pgp,rzf}}^o
\stackrel{M \rightarrow \infty}{\longrightarrow} 0 \;\;\; \mbox{with probability 1},
\end{equation}
where, for all users $g_k$, $\gamma_{g_k,{\rm pgp,rzf}}^o$ is a deterministic quantity that can be computed for any finite $M$ as
\begin{equation}
\gamma_{g_k,{\rm pgp,rzf}}^o = \frac{\frac{P}{S} \bar{\zeta}_g^2
(\bar{m}_{g}^o)^2}{\bar{\zeta}_g^2 \bar{\Upsilon}_{g,g}^o + (1 +
\sum_{g' \neq g} \bar{\zeta}_{g'}^2 \bar{\Upsilon}_{g,g'}^o) (1 +
\bar{m}_{g}^o)^2},
\end{equation}
where $\bar{\zeta}_g^2 = \frac{P/G}{\bar{\Gamma}_g^o}$ and the
quantities $\bar{m}_{g}^o$, $\bar{\Upsilon}_{g,g}^o$,
$\bar{\Upsilon}_{g,g'}^o$ and $\bar{\Gamma}_g^o$ are given by
\begin{eqnarray}
\label{fixed-pt-1-rzfbf-jsdm-sect} \bar{m}_{g}^o &=& \frac{1}{b'} \trace \left ( \bar{\Rm}_{g} \bar{\Tm}_g \right )\\
\label{fixed-pt-2-rzfbf-jsdm-sect} \bar{\Tm}_g &=& \left(
\frac{S'}{b'}  \frac{\bar{\Rm}_{g}}{1 + \bar{m}_{g}^o} +
\alpha \Id_{b'}\right)^{-1}\\
\bar{\Gamma}_g^o &=& \frac{1}{b'} \frac{P}{G} \frac{\bar{n}_{g}}{(1 +
\bar{m}_{g}^o)^2}\\
\bar{\Upsilon}_{g,g}^o &=& \frac{1}{b'} \frac{S' - 1}{S'}
\frac{P}{G} \frac{\bar{n}_{g,g}}{(1 + \bar{m}_{g}^o)^2} \\
\bar{\Upsilon}_{g,g'}^o &=& \frac{1}{b'} \frac{P}{G} \frac{
\bar{n}_{g',g}}{(1 + \bar{m}_{g'}^o)^2}\\
\bar{n}_{g} &=& \frac{\frac{1}{b'} \trace \left ( \bar{\Rm}_{g}
\bar{\Tm}_g \Bm_{g}^\herm\Bm_{g} \bar{\Tm}_g  \right )}{1 -
\frac{\frac{S'}{b'} \trace \left ( \bar{\Rm}_{g}
\bar{\Tm}_g \bar{\Rm}_{g} \bar{\Tm}_g  \right )}{b' (1 + \bar{m}_{g}^o)^2}}\\
\bar{n}_{g,g} &=& \frac{\frac{1}{b'} \trace \left ( \bar{\Rm}_{g}
\bar{\Tm}_g \bar{\Rm}_{g} \bar{\Tm}_g  \right )}{1 - \frac{\frac{S'}{b'}
\trace\left ( \bar{\Rm}_{g} \bar{\Tm}_g \bar{\Rm}_{g} \bar{\Tm}_g  \right )}{b' (1 +
\bar{m}_{g}^o)^2}}\\
\bar{n}_{g',g} &=& \frac{\frac{1}{b'} \trace \left ( \bar{\Rm}_{g'}
\bar{\Tm}_{g'} \Bm_{g'}^\herm\Rm_{g}\Bm_{g'} \bar{\Tm}_{g'}  \right )}{1 -
\frac{\frac{S'}{b'} \trace \left ( \bar{\Rm}_{g'} \bar{\Tm}_{g'}
\bar{\Rm}_{g'} \bar{\Tm}_{g'}  \right )}{b' (1 + \bar{m}_{g'}^o)^2}}
\end{eqnarray}

\subsection{Validation of the asymptotic analysis}
\label{subsec:res-jsdm}

In this section we present some numerical examples focusing on
the case when the tall unitary condition is not satisfied, and we discuss the choice
of  the effective rank parameter $r^\star$ in the approximated BD for PGP
(more in general, the parameters $\{r^\star_g\}$, for an asymmetric case).
We also compare the results obtained via the method of deterministic equivalents with finite-dimensional Monte Carlo
simulations, in order to give an idea on the method accuracy.~\footnote{Precise statements on the order of convergence
with respect to $M$ of the actual finite dimensional SINRs to their deterministic equivalents are given in \cite{debbah2012}.}

In the following examples, the BS is equipped with a uniform circular array
with $M = 100$ isotropic antenna elements equally spaced on a circle of radius $\lambda D$,
for $D = \frac{0.5}{\sqrt{(1 - \cos(2\pi/M))^2 + \sin(2\pi/M)^2}}$, resulting in the
minimum distance between antenna elements equal to $\frac{\lambda}{2}$.
Users form $G = 6$ symmetric groups,  with AS $\Delta = 15^o$ and azimuth AoA
$\theta_g = -\pi + \Delta + (g-1)\frac{2\pi}{G}$  for $g = 1,\ldots, G$.
The user channel correlation is obtained according to (\ref{eq:SM-4}).
For the system geometry defined above, the transmit covariance
matrix for each group has rank $r = 21$. However, half of the non-zero eigenvalues are extremely small, yielding an effective
rank $r^\star = 11$. Somehow arbitrarily, we fixed to serve $S' = 5$ data streams per group,
so that the total number of users being served is $S = S' G = 30$, and chose $b' = 10$.

%

\begin{figure}
\centering \subfigure[$r^\star = 6$]{
  \includegraphics[width=8cm]{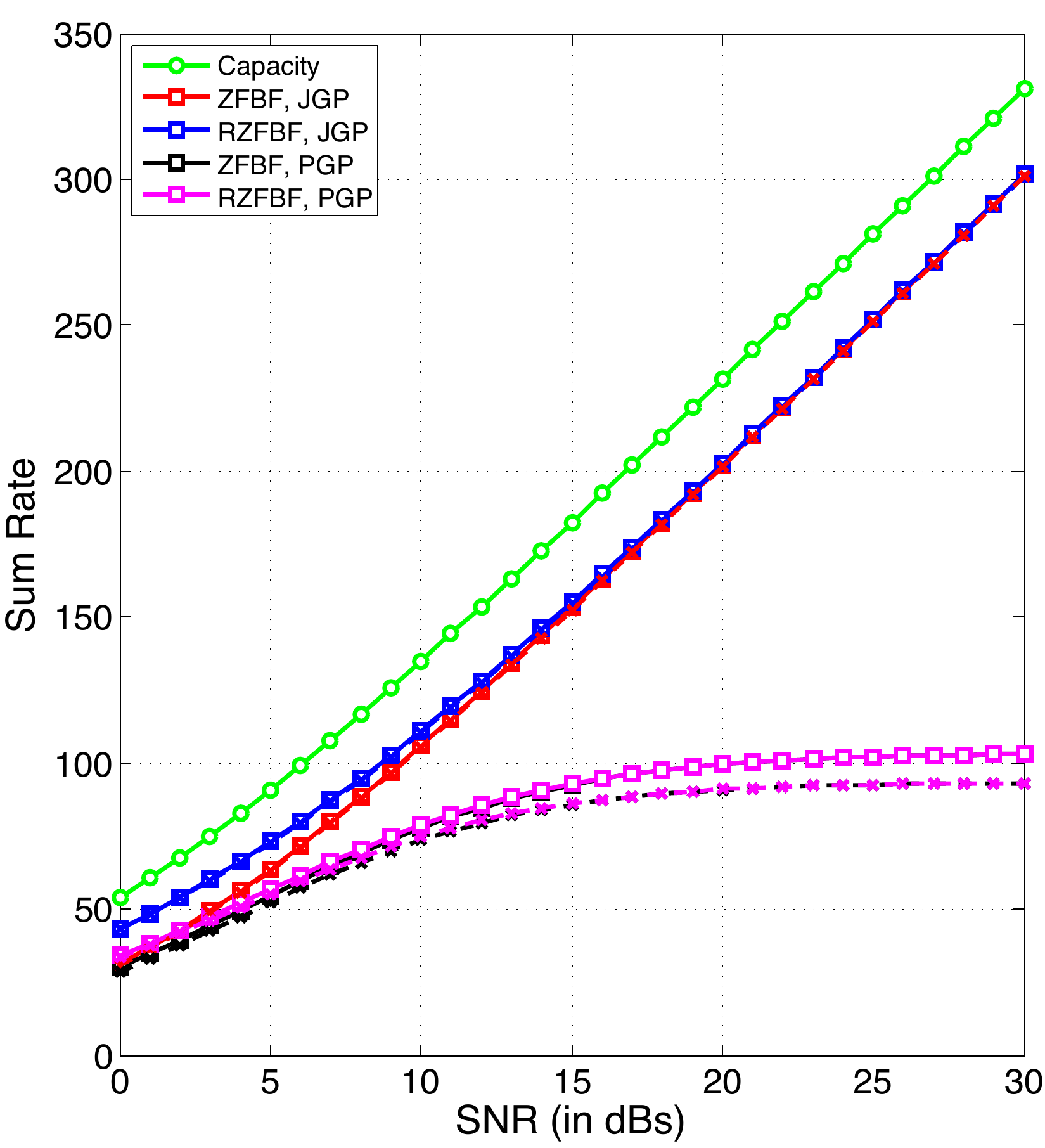}
  \label{fig:s'-5-r*-6}
  }
  \subfigure[$r^\star = 11$]{
  \includegraphics[width=8cm]{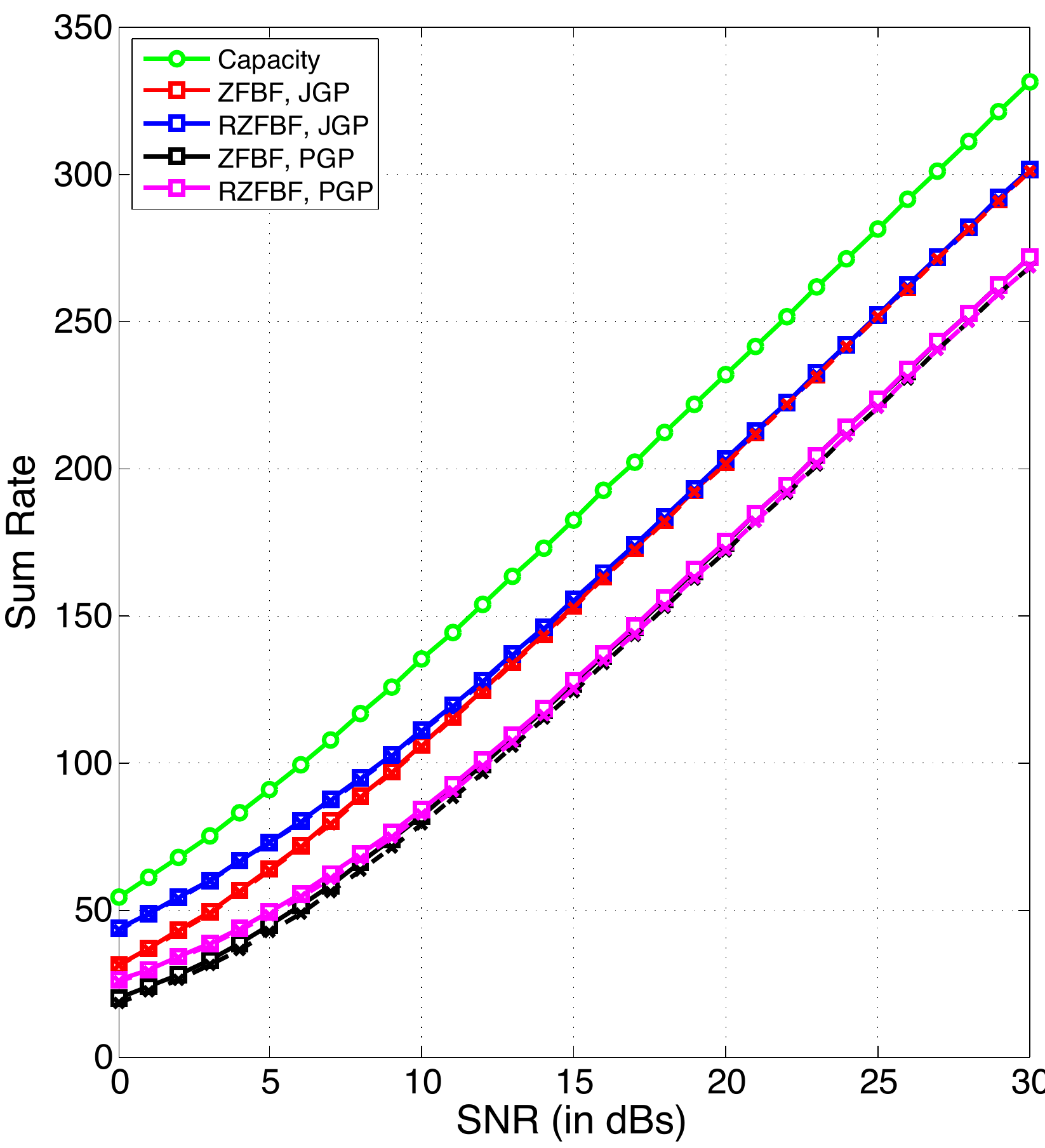}
  \label{fig:s'-5-r*-12}
  }
  \caption{Comparison of sum spectral efficiency (bit/s/Hz) vs. SNR (dB) for JSDM with their corresponding deterministic equivalents.
  ``JGP'' denotes JSDM with \emph{joint group processing} and ``PGP'' denotes JSDM with \emph{per-group processing}.}
  \label{fig:jsdm-results}
\end{figure}

Figs. \ref{fig:s'-5-r*-6} and \ref{fig:s'-5-r*-12} show the
performance of the JSDM schemes when the pre-beamforming matrix is
designed according to the approximate BD method described in Section
\ref{sec:BD}, choosing $r^\star = 6$ and $r^\star = 12$,
respectively. Given the noise unit variance normalization, we have
that ${\rm SNR} = P$. The solid ``squares'' are obtained through
simulations and the dotted ``x'' are obtained using the
corresponding deterministic equivalent approximations. The
regularization parameter is fixed to $\alpha = \frac{S}{bP}$ for
both JGP and PGP. The performance of JSDM with JGP in Figs.
\ref{fig:s'-5-r*-6} and \ref{fig:s'-5-r*-12} is identical, owing to
the fact that we use eigen-beamforming with $\Bm_g = \Um_g$,
independent of $r^\star$. For the sake of comparison, the sum
capacity of the MIMO BC channel with full CSIT (see (\ref{eq:SM-3}))
is also shown (solid ``circles'' in green), obtained by the
iterative waterfiling approach of \cite{yu2006sum}.

\begin{rem} \label{choice-of-rstar-remark}
By choosing $r^\star$ too small, such that significant eigenmodes are not taken into account
by the approximate BD pre-beamforming matrix, the resulting inter-group interference is
large and the performance of PGP is severely interference limited (e.g., Fig. \ref{fig:s'-5-r*-6}).
Instead, by choosing $r^\star$ large enough, in order to include all significant eigenmodes,
the performance of PGP does not show a noticeable  interference limited behavior over a wide range of SNR.
This is the case of Fig. \ref{fig:s'-5-r*-12}, where we
chose $r^\star = 12$ and the channel covariance matrix has rank $r = 21$, but only $11$ significant eigenvalues.
As a matter of fact, the PGP rate curves of Fig. \ref{fig:s'-5-r*-12} will eventually flatten, but this happens at
extremely large SNR, irrelevant for practical applications. This example shows that $r^\star$ should always be chosen in order
to include all strongest eigenmodes. However, making $r^\star = r$ is generally not a good choice
since many eigenmodes may be very close to zero (as in this example) and therefore including them in the count of $r^\star$ yields a dimensionality
bottleneck without any real benefit in terms of inter-group interference (recall that $r^\star G \leq M$, therefore if $r^\star$ is large we may have to decrease $G$,
i.e., serve less groups in parallel). We conclude that the choice of the effective rank $r^\star$ should be carefully optimized,
depending on  the specific channel covariance eigenvalue distribution.
\hfill $\lozenge$
\end{rem}


\section{Downlink training and noisy CSIT} \label{sec:nonperfect-csi}

In this section, we evaluate the impact of noisy CSIT by including the fact that the
effective channels are estimated by the UTs from the downlink training phase.
In the vast literature dedicated to CSIT feedback (see for example
\cite{Caire-Jindal-Kobayashi-Ravindran-TIT10} and references therein), methods that achieve the estimated channel
Mean-Square Error (MSE) decreases as $O(1/P^\beta)$ for some $\beta \geq 1$, even in the presence of
channel feedback noise and errors, are well-known.  In contrast, the MSE due to estimation from the downlink training phase
decreases at best as $O(1/P)$. In fact, this is given by the high-SNR behavior of the MMSE for a Gaussian signal (the channel vectors)
in Gaussian noise. If the CSIT feedback scheme is designed to achieve exponent $\beta > 1$ and the channel SNR
is sufficiently large, the feedback error is negligible with respect to the downlink estimation error \cite{Caire-Jindal-Kobayashi-Ravindran-TIT10}.
Hence, for simplicity, we consider the optimistic situation of ideal and delay-free CSIT feedback, and focus only on the effect of the downlink
channel estimation error and dimensionality penalty factor of the training phase (a similar approach is
followed in \cite{Huh-Tulino-Caire-TITsubmit}).

For brevity, we focus only on the case of
PGP.\footnote{Analogous results can be obtained for the case of JGP,
but these are practically less interesting since JGP requires
typically too large training and feedback overhead in FDD systems.}
From Section \ref{sec:PGP}, the channel covariance matrix for a user
$g_k$ is given by $\bar{\Rm}_{g} = \Bm_g^\herm \Rm_g \Bm_g$. In
order to estimate the effective channel vector $\textsf{\hv}_{g_k}$,
i.e., the column of the effective channel matrix $\textsf{\pH}_g$
corresponding to user $g_k$,  the BS sends unitary training
sequences of length $b'$, in parallel over the $b'$ virtual inputs
of the pre-beamforming of each group $g$. Hence, the training
phase with PGP spans $b'$ symbols. The UTs in each group make use of
linear MMSE estimation, which is the optimal estimator for
minimizing the MSE since the observation at each user and the
channel vector are conditionally jointly Gaussian given the training
sequences. The MMSE channel estimates are fed back to the BS and are
used to compute the linear precoders $\{\Pm_g\}$. Assuming that in each
coherence block of $T$ symbols the training phase makes use of $b'$
symbols, and the remaining $T - b'$ symbols are available for
downlink data transmission, it follows that the spectral efficiency
must be scaled by the dimensionality penalty factor $\max\{ 1 - b'/T , 0\}$.

We consider a scheme where a scaled unitary training matrix
$\Xm_{\rm tr}$ of dimension $b' \times b'$ is sent, simultaneously, to all groups in
the common downlink training phase. The corresponding received signal
at group $g$ receivers is given by
\begin{align} \label{eq:training1}
   \Ym_g = \textsf{\pH}_g^\herm  \Xm_{\rm tr} +  \sum_{g' \neq g}  {\pH}_g{}^\herm \pB_{g'}  \Xm_{\rm tr}   +  \Zm_g.
\end{align}
Multiplying from the right by $\Xm_{\rm tr}^\herm$ and using the fact that, by design, $\Xm_{\rm tr}\Xm_{\rm tr}^\herm = \rho_{\rm tr} \Id_{b'}$ where
$\rho_{\rm tr}$ is the power allocated to training, we obtain
\begin{align} \label{eq:training2}
   \Ym_g \Xm_{\rm tr}^\herm = \rho_{\rm tr} \textsf{\pH}_g^\herm +  \rho_{\rm tr}  \sum_{g' \neq g}  {\pH}_g{}^\herm \pB_{g'}    +  \Zm_g \Xm_{\rm tr}^\herm.
\end{align}
Extracting the $g_k$-th row, dividing by $\sqrt{\rho_{\rm tr}}$, using the fact that $\Zm_g \Xm_{\rm tr}^\herm$ has i.i.d. entries $\sim \Cc\Nc(0,\rho_{\rm tr})$
and taking Hermitian transpose of everything, we obtain the noisy observation for estimating  the $g_k$-th effective channel vector in the form
\begin{align} \label{eq:training3}
\widetilde{\textsf{\hv}}_{g_k}   = \sqrt{\rho_{\rm tr}} \textsf{\hv}_{g_k} + \sqrt{\rho_{\rm tr}} \left ( \sum_{g' \neq g} \Bm_{g'}^\herm \right ) \hv_{g_k}
+ \widetilde{\zv}_{g_k},
\end{align}
where $\widetilde{\zv}_{g_k} \sim \Cc\Nc(\zerov, \Id_{b'})$.
The MMSE estimator for $\textsf{\hv}_{g_k}$ based on  (\ref{eq:training3}) is given by
\begin{align}
\widehat{\textsf{\hv}}_{g_k} & = \EE\left [\textsf{\hv}_{g_k} \widetilde{\textsf{\hv}}_{g_k}^\herm \right ] \EE \left [ \widetilde{\textsf{\hv}}_{g_k} \widetilde{\textsf{\hv}}_{g_k}^\herm \right ]^{-1} \widetilde{\textsf{\hv}}_{g_k} \nonumber \\
& = \sqrt{\rho_{\rm tr}} \left [ \Bm_g^\herm \Rm_g \sum_{g'=1}^G \Bm_{g'} \right ]
\left [ \rho_{\rm tr} \sum_{g',g''=1}^G \Bm_{g'}^\herm \Rm_g \Bm_{g''}  + \Id_{b'} \right ]^{-1} \widetilde{\textsf{\hv}}_{g_k} \nonumber \\
& = \frac{1}{\sqrt{\rho_{\rm tr}}} \left ( \Mm_g \tilde{\Rm}_g \Om^\transp \right ) \left [ \Om \tilde{\Rm}_g \Om^\transp + \frac{1}{\rho_{\rm tr}} \Id_{b'} \right ]^{-1} \widetilde{\textsf{\hv}}_{g_k} \label{messy-mmse-estimator}
\end{align}
where we used the fact that $\textsf{\hv}_{g_k} = \Bm_g^\herm \hv_{g_k}$, where $\tilde{\Rm}_g$ is defined in (\ref{ch-cov-jsdm}) and we introduced
the $b' \times b$ block matrices
\begin{eqnarray*}
\Mm_g & = & [\zerov, \ldots, \zerov, \underbrace{\Id_{b'}}_{{\rm block} \; g}, \zerov, \ldots, \zerov] \\
\Om     & = & [\Id_{b'}, \Id_{b'}, \ldots, \Id_{b'} ].
\end{eqnarray*}
Notice that in the case of perfect BD we have that $\Rm_g \Bm_{g'} = \zerov$ for $g' \neq g$. Therefore,
(\ref{eq:training3}) and (\ref{messy-mmse-estimator}) reduce to
\begin{align} \label{eq:training4}
\widetilde{\textsf{\hv}}_{g_k}   = \sqrt{\rho_{\rm tr}} \textsf{\hv}_{g_k} + \widetilde{\zv}_{g_k},
\end{align}
and
\begin{align} \label{eq:training5}
\widehat{\textsf{\hv}}_{g_k} & =  \frac{1}{\sqrt{\rho_{\rm tr}}} \bar{\Rm}_g \left [ \bar{\Rm}_g + \frac{1}{\rho_{\rm tr}} \Id_{b'} \right ]^{-1} \widetilde{\textsf{\hv}}_{g_k}
\end{align}
respectively, where we recall the definition $\bar{\Rm}_g = \Bm_g^\herm \Rm_g \Bm_g$.

For this channel estimation scheme, the deterministic equivalent approximation of the SINR terms for
RZFBF and ZFBF precoding can be obtained following \cite{debbah2012,hoydis2011massive},
the approach of which can be directly applied to our case, and using the well-known MMSE decomposition
\begin{equation}
\label{eq:MMSE-decomp} \textsf{\hv}_{g_k} = \widehat{\textsf{\hv}}_{g_k} + \widehat{\textsf{\ev}}_{g_k},
\end{equation}
with $\EE[\widehat{\textsf{\hv}}_{g_k} \widehat{\textsf{\hv}}_{g_k}^\herm] = \widehat{\bar{\Rm}}_{g}$ and MMSE covariance matrix $\EE[\widehat{\textsf{\ev}}_{g_k}
\widehat{\textsf{\ev}}_{g_k}^\herm] = \bar{\Rm}_{g} - \widehat{\bar{\Rm}}_{g}$.
For completeness, the fixed-point equations leading to the deterministic equivalent SINR approximation for PGP with noisy CSIT
are given in Appendix \ref{sec:determ-equiv-nonideal-csi}. Eventually, the achievable rate of user $g_k$ is approximated by
\begin{eqnarray}
R_{g_k,{\rm pgp,csit}} & = & \max\left \{ 1 - \frac{b'}{T}, 0 \right \} \times \log(1 + \widehat{\gamma}_{g_k,{\rm pgp,csit}}^o ),
\end{eqnarray}
where $\widehat{\gamma}_{g_k,{\rm pgp,csit}}^o$ indicates either
$\widehat{\gamma}_{g_k,{\rm pgp,rzf,csit}}^o$ or  $\widehat{\gamma}_{g_k,{\rm pgp,zf,csit}}^o$, as detailed in
Appendix \ref{sec:determ-equiv-nonideal-csi}.

\begin{rem}
Assuming that, as $M \rightarrow \infty$, the other system dimensions
$r^\star, S$ and $b$ also go to infinity linearly with $M$, the achievable rate approximation error converges
to zero almost surely as $M \rightarrow \infty$. However, the dimensionality factor $\max\{1 - b'/T,0\}$ is equal to
zero for $b' \geq T$. Hence, in order to obtain mathematically meaningful results we assume that
also the coherence block length $T$ grows linearly with $M$, and we define the factor $\tau = b'/T$ as the {\em dimensionality crowding factor}
of the channel. In practice, this means that the method is valid in the regime of $b'$ large, but still significantly smaller than $T$.
\hfill $\lozenge$
\end{rem}

\subsection{Results with downlink channel estimation} \label{sec:prime-tradeoff}

\begin{figure}
\centering \subfigure[$S' = 4$]{
  \includegraphics[width=8cm,bb = 50 100 600 600]{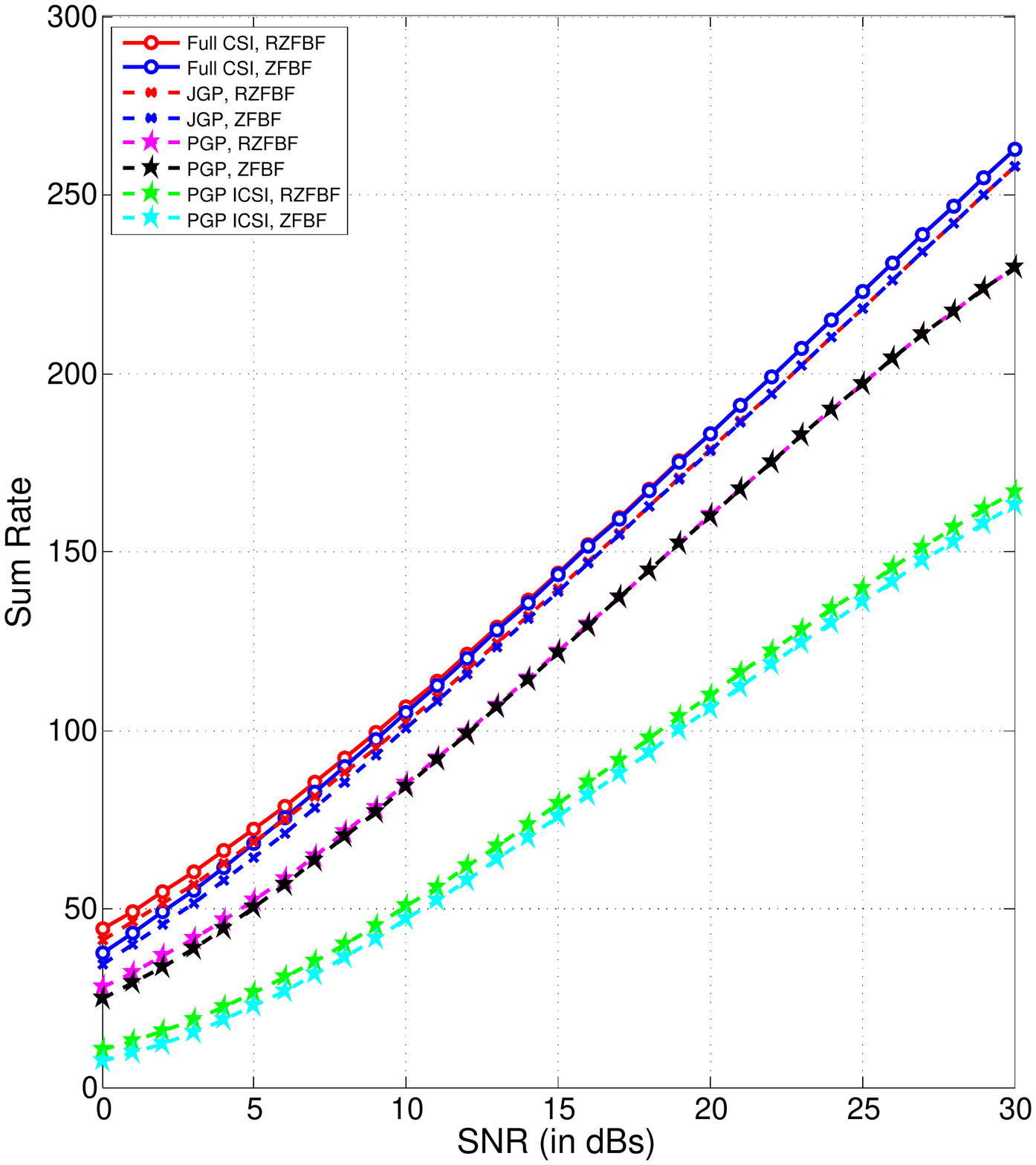}
  }
  \subfigure[$S' = 8$]{
  \includegraphics[width=8cm,bb = 50 100 600 600]{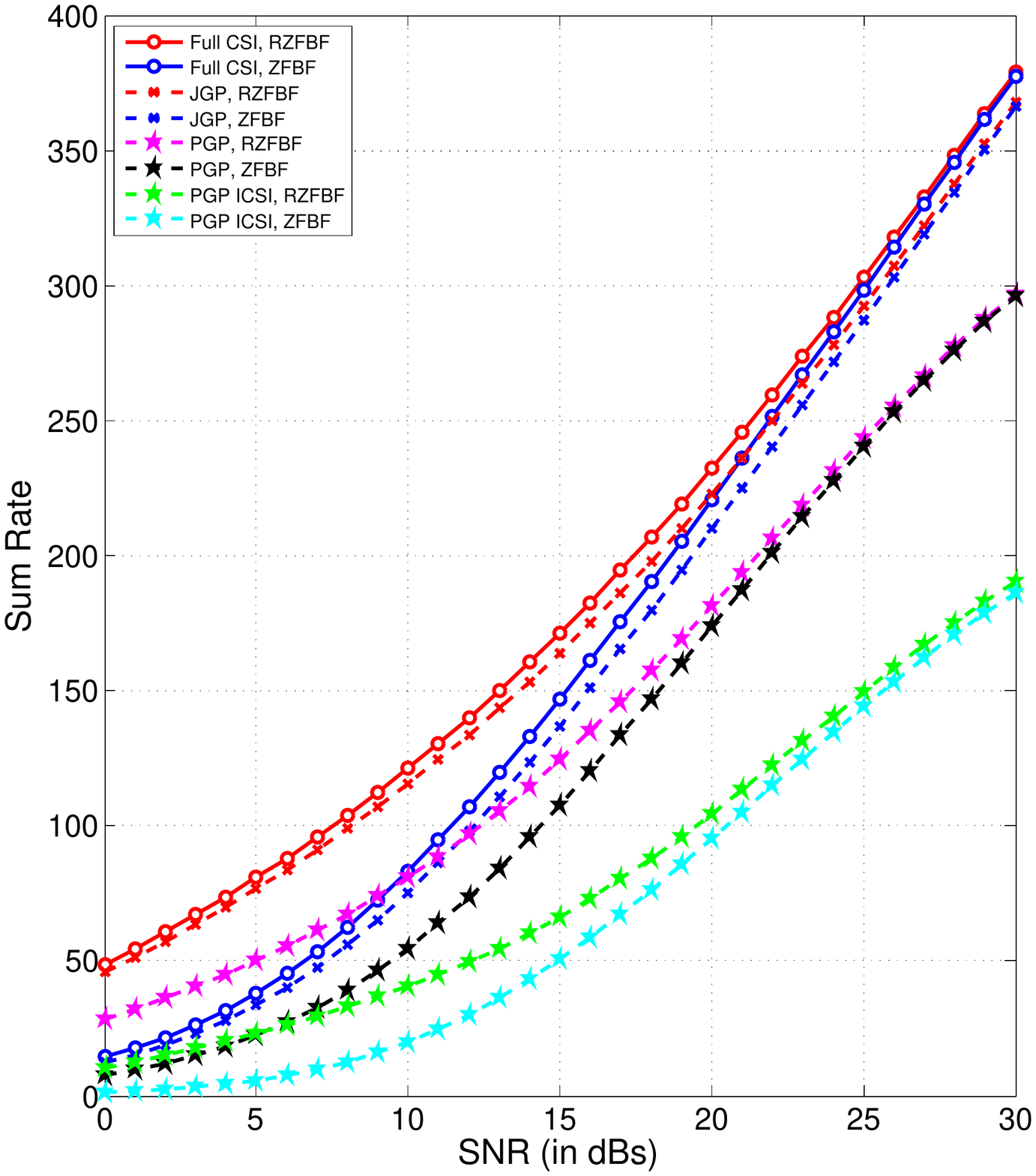}
  }
  \caption{Sum spectral efficiency (bit/s/Hz) vs. SNR (dB) for JSDM (computed via deterministic equivalents) with $r^\star=11$,
  for $S' = 4$ and $S' = 8$. The coherence block length is $T = 40$. The ``green'' and ``cyan'' curves denote the results for imperfect
  CSIT with optimized choice of $b'$. ``JGP'' denotes JSDM with \emph{joint group processing} and ``PGP'' denotes
  JSDM with \emph{per-group processing}.}
  \label{fig:icsi-results}
\end{figure}

We demonstrate the effect of noisy CSIT on the performance of RZFBF
and ZFBF in Fig. \ref{fig:icsi-results}, for the same antenna
configuration of Section \ref{subsec:res-jsdm} with $r^\star = 11$,
for $S' = 4$ and $S' = 8$ streams per group. For the sake of
comparison, the solid ``red'' (``blue'') curve denotes the sum
spectral efficiency achieved by RZFBF (ZFBF) with full noiseless
CSIT, i.e., by computing the precoding matrix in one step, directly
from the instantaneous channel matrix $\underline{\Hm}$. The dotted
lines represent the performance of JSDM for JGP with
eigen-beamforming and noiseless CSIT (i.e., perfect knowledge of the
effective channel $\underline{\textsf{\pH}}$). The ``magenta''
(``black'') curves denote the sum spectral efficiency for JSDM with
PGP and approximate BD, also in the case of noiseless CSIT. Finally,
the ``green'' (``cyan'') curves denote the achievable sum spectral
efficiency for JSDM with PGP and noisy CSIT, obtained by downlink training and
MMSE estimation as explained above. These curves are
obtained by optimizing the parameter $b'$, for given $S'$, $r^\star$
and SNR. Since a set of training sequences is sent simultaneously to all groups, the training power is given by $\rho_{\rm tr} =
\frac{P}{G}$, such that the total sum power constraint is preserved
also during the training phase.

\begin{rem} \label{choice-of-bprime-remark}
We examine now the optimization of the parameter $b'$ for fixed
target $S'$, in the case of downlink training and noisy CSIT. Having
fixed $r^\star$ as discussed in Remark \ref{choice-of-rstar-remark},
and assuming $0 \leq S' \leq b' \leq M - r^\star(G-1)$, for each
value of SNR and given JSDM precoding scheme there is an optimal
choice of $b'$. For example, Fig. \ref{fig:optimal-b-8} shows the
dependency of the sum spectral efficiency of JSDM with PGP with
respect to  $b'$ for $S' = 8$ and $\SNR = 10$ and 30 dB. We notice
that the sum spectral efficiency  including channel estimation is
not monotonically increasing with $b'$. In fact, letting $b'$ large
yields better conditioned effective channel matrices, but incurs a
larger dimensionality cost of the downlink  training phase. The
tension between these two issues yields a non-trivial choice for the
optimal value of $b'$ maximizing the system spectral efficiency.
Similar trends can be observed for different values of $S'$ and
different values of SNR. \hfill $\lozenge$
\end{rem}

\begin{figure}
\centering \subfigure[$S' = 8$, SNR = 10 dB]{
  \includegraphics[width=8cm,bb = 50 100 600 600]{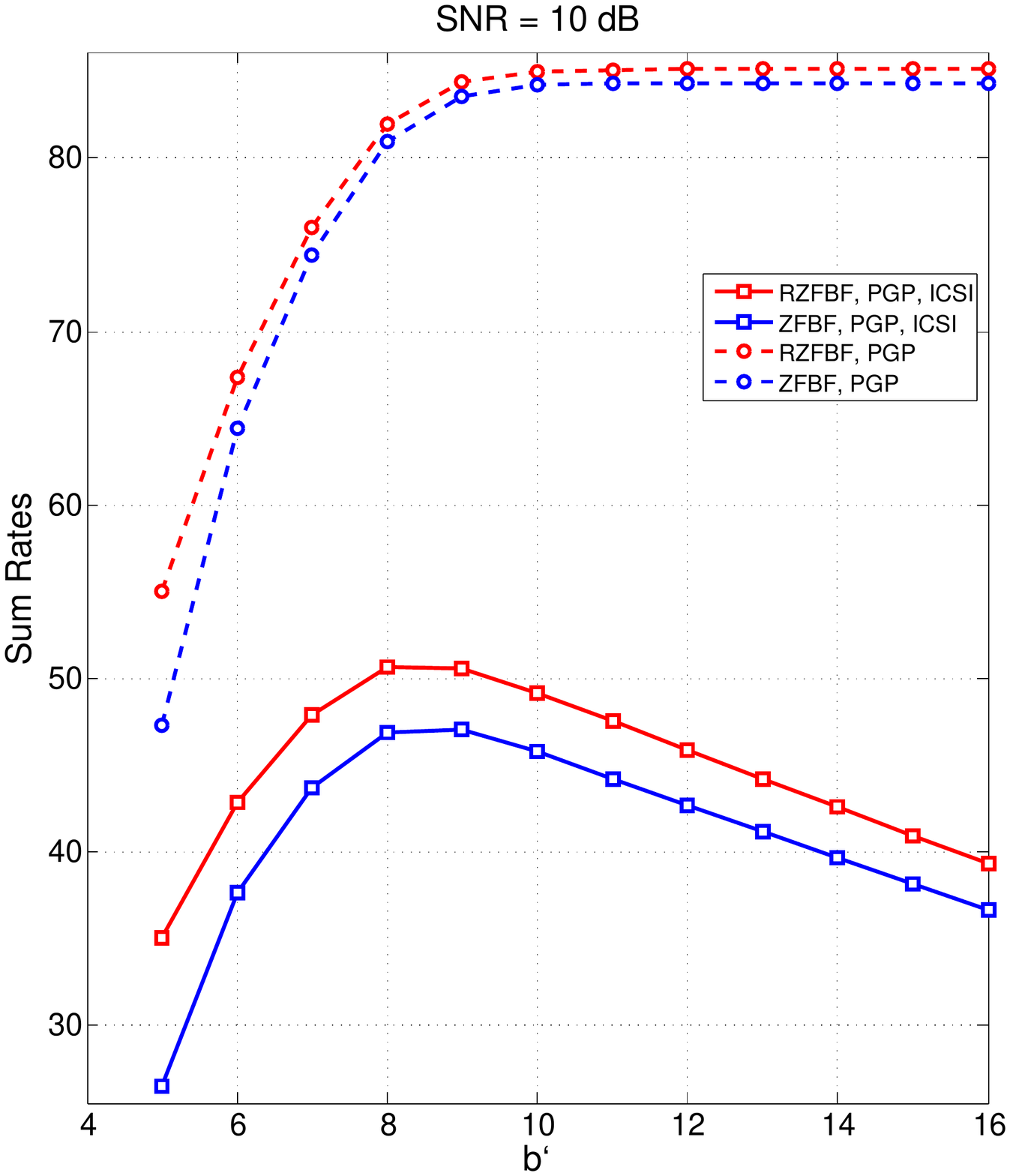}
  }
  \subfigure[$S' = 8$, SNR = 30 dB]{
  \includegraphics[width=8cm,bb = 50 100 600 600]{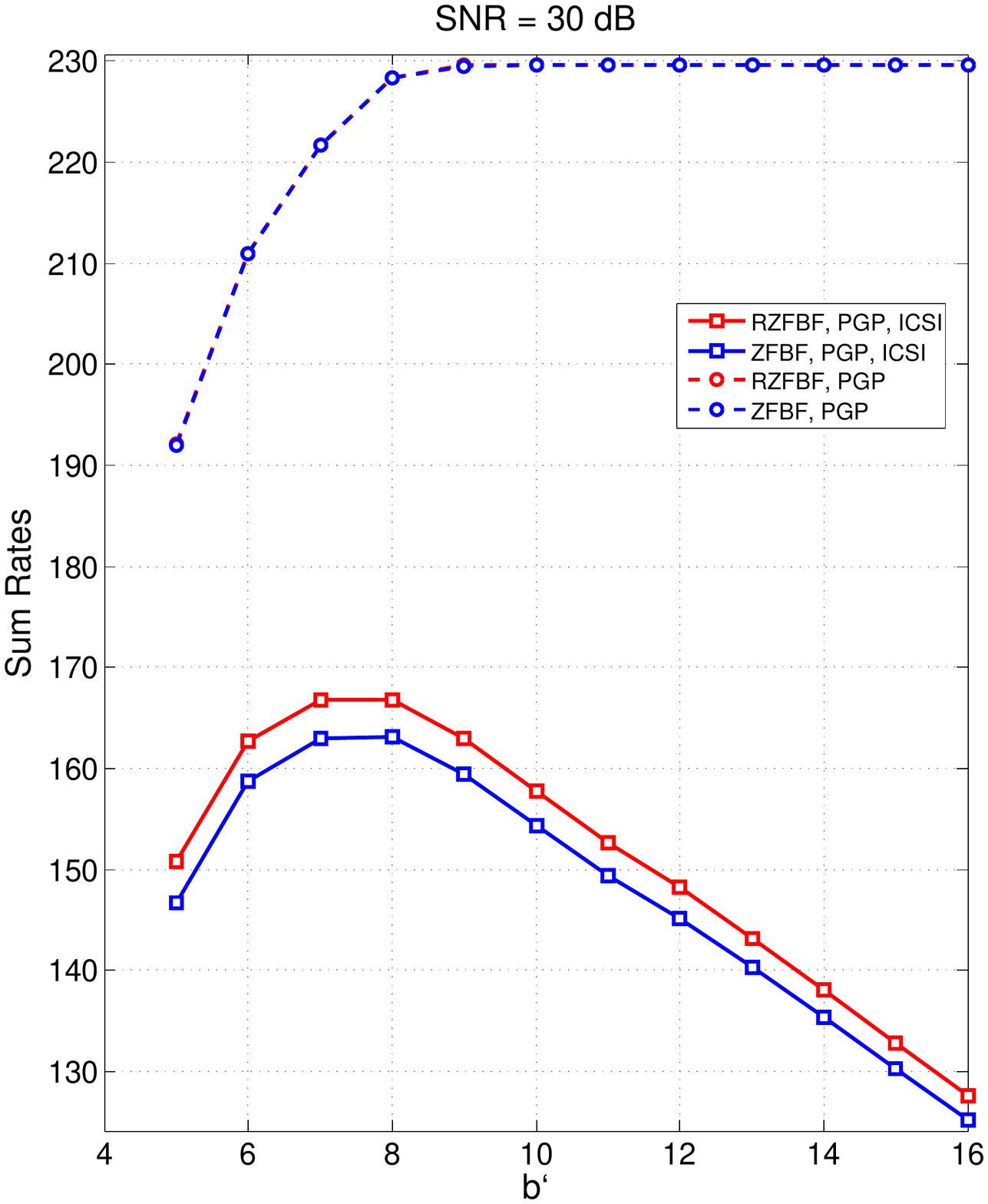}
  }
  \caption{Sum spectral efficiency (bit/s/Hz) vs. $b'$ for JSDM with $r^\star=11$,
  for $S' = 8$ (computed via deterministic equivalents).
  The coherence block length $T = 40$. The ``dashed'' curves denote the results for PGP with perfect CSIT,
  and the ``solid'' lines denote the same for imperfect CSIT.}
  \label{fig:optimal-b-8}
\end{figure}

\begin{rem} \label{choice-of-Sprime-remark}
Having chosen $b'$, we focus now on choice of optimal $S'$. This
depends heavily on the precoding scheme and the
operating SNR. For a given operating SNR, there is approximately a linear dependence between the optimal $S'$ and $b'$
for both the RZFBF and the ZFBF precoders considered here.
This linear dependence can be characterized by a single parameter, namely, the slope of the line relating
the optimal $S'$ and $b'$. In Fig.~\ref{fig:slope-SNR} we have plotted this slope versus SNR.
It can be seen that for RZFBF, at low values of SNR, the choice $S' = b'$ (slope equal to 1) is optimal.
In contrast, for ZFBF it is better to serve some $S' < b'$ number of users.  As the SNR increases, the ZFBF slope increases and
approaches that of the same slope of RZFBF at high SNR.
\end{rem}

\begin{figure}
\centering
  \includegraphics[width=10cm]{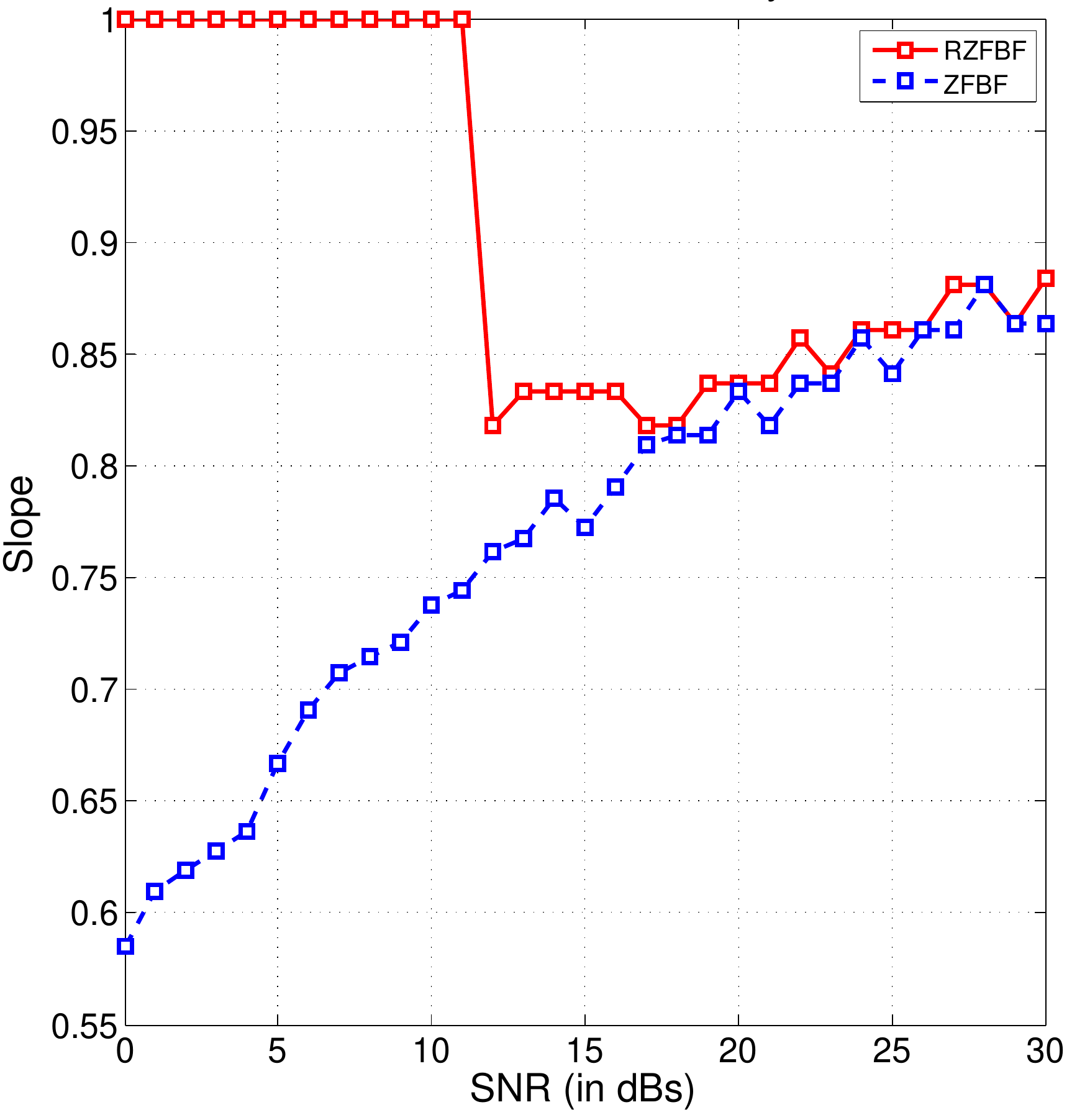}
  \caption{Ratio $S'/b'$ (slope) for the optimized  $S'$ and $b'$ versus the channel SNR for different precoders.}
  \label{fig:slope-SNR}
\end{figure}

\section{Uniform linear arrays: eigenvalues and eigenvectors}
\label{sec:dft-prebeamforming}

In this section we consider the antenna correlation model (\ref{eq:SM-4}) for the special but important
case of a Uniform Linear Array (ULA) of large dimension ($M \gg 1$), and obtain important insight on the
behavior of the normalized asymptotic rank $\rho = \lim_{M \rightarrow \infty} \frac{r}{M}$
and of the eigenvectors $\Um$ of the covariance matrix $\Rm$.
We consider a 120 deg sector, obtained by using directional radiating elements, and
assume that the sector is centered around the x-axis ($\alpha = 0$ azimuth angle),
and that no energy is received for angles $\alpha \notin [-\pi/3, \pi/3]$.
A ULA formed by $M$ such directional radiating elements is placed at the origin along
the y-axis. Denoting by $\lambda D$ the spacing of the antenna elements, the covariance matrix
of the channel for a user at AoA $\theta$ and AS $\Delta$ according to the model of Section \ref{sec:channel-model}
is given by the Toeplitz form
\begin{equation}  \label{correlation-matrix-ULA}
[\Rm]_{m,p} = \frac{1}{2\Delta} \int_{-\Delta+\theta}^{\Delta+\theta}   e^{-j2\pi D (m-p) \sin(\alpha)} d\alpha
\end{equation}
for $m,p \in \{0,1,\ldots, M-1\}$.
In order to characterize eigenvalues and eigenvectors of $\Rm$ with respect to $D, \Delta, \theta$, for large $M$,
we resort to the well-known results of \cite{gray2006toeplitz,grenander1984toeplitz}.

From \cite{grenander1984toeplitz}, we recall the following
fundamental result. Let $S(\xi)$ be a uniformly bounded absolutely
integrable function over $\xi \in [-1/2,1/2]$, i.e.,
\[ \int_{-1/2}^{1/2} |S(\xi)| d\xi < \infty, \;\; \kappa_1 \leq S(\xi) \leq \kappa_2, \]
where the bounds hold for all $\xi \in [-1/2,1/2]$ up to a set of measure zero.
Assume that we can write the sequence $r_m = [\Rm]_{\ell,\ell-m}$ as the
inverse discrete-time  Fourier transform of $S(\xi)$, i.e.,
\begin{equation} \label{PSD}
r_m = \int_{-1/2}^{1/2} S(\xi) e^{j2\pi \xi m} d\xi.
\end{equation}
Then, the Toeplitz matrix $\Rm$ can be approximated by the circulant matrix $\Cm$ defined by its first column with $m$-th element
\begin{equation} \label{circ-approx}
c_m = \left \{ \begin{array}{ll}
r_m + r_{m-M} & \mbox{for} \;\; m = 1,  \ldots, M - 1\\
r_0 & \mbox{for} \;\; m = 0 \end{array} \right . ,
\end{equation}
where the approximation holds in the following sense:

\begin{fact} \label{eigenvalue-approx}
The set of eigenvalues $\{\lambda_m(\Rm)\}$, $\{\lambda_m(\Cm)\}$ and the set of uniformly spaced samples $\{S(m/M): m = 0,\ldots, M-1\}$
are asymptotically {\em equally distributed}, i.e.,  for any continuous function $f(x)$ defined over $[\kappa_1, \kappa_2]$, we have
\begin{equation}
\lim_{M \rightarrow \infty} \frac{1}{M} \sum_{m=0}^{M-1} f(\lambda_m(\Rm)) = \lim_{M \rightarrow \infty} \frac{1}{M} \sum_{m=0}^{M-1} f(\lambda_m(\Cm)) =
\int_{-1/2}^{1/2}  f(S(\xi))  d\xi.
\end{equation}
\hfill \QED
\end{fact}

\begin{fact} \label{eigenspace-approx}
The eigenvectors of $\Rm$ are approximated by the eigenvectors of
$\Cm$ in the following eigenspace approximation sense. Define the
asymptotic eigenvalue cumulative distribution function (CDF) of the
eigenvalues of $\Rm$ to be the right-continuous non-decreasing
function $F(\lambda)$ such that  $F(\lambda) = \int_{S(\xi) \leq
\lambda} d\xi$ for any point of continuity $\kappa_1 \leq \lambda
\leq \kappa_2$. Let $\lambda_0(\Rm) \leq \ldots, \leq
\lambda_{M-1}(\Rm)$ and $\lambda_0(\Cm) \leq \ldots, \leq
\lambda_{M-1}(\Cm)$ denote the set of ordered eigenvalues of $\Rm$
and $\Cm$, and let $\Um = [\uv_0, \ldots, \uv_{M-1}]$ and $\Fm =
[\fv_0, \ldots, \fv_{M-1}]$ denote the corresponding
eigenvectors.~\footnote{Notice that in the channel model defined in
Section \ref{sec:channel-model} we defined $\Um$ of dimensions $M
\times r$ to be the matrix of eigenvectors corresponding to the
non-zero eigenvalues of $\Rm$. In the statement of this result,
instead, $\Um$ denotes the whole $M \times M$ matrix of
eigenvectors, including the non-unique eigenvectors forming a
unitary basis for the nullspace of $\Rm$, in the case $r < M$.} For
any interval $[a, b] \subseteq [\kappa_1,\kappa_2]$ such that
$F(\lambda)$ is continuous on $[a,b]$, consider the eigenvalues
index sets $\Ic_{[a,b]} = \{m :  \lambda_m(\Rm) \in [a,b]\}$ and
$\Jc_{[a,b]} = \{m :  \lambda_m(\Cm) \in [a,b]\}$, and define
$\Um_{[a,b]} = (\uv_m : m \in \Ic_{[a,b]})$ and  $\Fm_{[a,b]} =
(\fv_m : m \in \Jc_{[a,b]})$ be the submatrices of $\Um$ and $\Fm$
formed by the columns whose indices belong to the sets $\Ic_{[a,b]}$
and $\Jc_{[a,b]}$, respectively. Then, the eigenvectors of $\Cm$
approximate the eigenvectors of $\Rm$ in the sense that
\begin{equation}
\lim_{M \rightarrow \infty} \frac{1}{M} \left \| \Um_{[a,b]} \Um_{[a,b]}^\herm  - \Fm_{[a,b]} \Fm_{[a,b]}^\herm \right \|_F^2 = 0.
\end{equation}
\hfill \QED
\end{fact}

A well-known property of circulant matrices \cite{gray2006toeplitz}
is that their eigenvectors form a unitary DFT matrix, i.e., the
matrix whose $(\ell,m)$-th element is given by $[\Fm]_{\ell,m} =
\frac{e^{-j2\pi \ell m/M}}{\sqrt{M}}$. This has an important
consequence for JSDM with large ULAs: in the regime of large $M$
where the Toeplitz channel correlation matrix $\Rm$ is well
approximated by its circulant version $\Cm$,  we can approximate
$\Um$, the tall unitary matrix of the channel covariance
eigenvectors, with a submatrix of $\Fm$, formed by a selection of
columns of $\Fm$. Hence, we can design the pre-beamforming stage of
JSDM by replacing $\Um$ with its DFT approximation, avoiding the
need of a precise estimation of the actual channel covariance
matrix.  In order to understand how to select the columns of $\Fm$,
we need to gain more insight into the asymptotic behavior of the
eigenvalues of $\Rm$.

For $r_m = [\Rm]_{\ell,\ell-m}$ with $[\Rm]_{m,p}$ given by (\ref{correlation-matrix-ULA}),
and $\Cm$ defined as in (\ref{circ-approx}),
the eigenvalues $\{\lambda_k(\Cm)\}$ can be given explicitly for any finite $M$ as follows:
\begin{eqnarray} \label{eigv-circ}
\lambda_k(\Cm)
& = & \sum_{m=0}^{M-1} c_m e^{-j\frac{2\pi}{M} m k} \nonumber \\
& = & r_0 + \sum_{m=1}^{M-1} [r_m + r_{m-M}] e^{-j\frac{2\pi}{M} m k} \nonumber \\
& = & r_0 + \sum_{m=1}^{M-1} r_m e^{-j\frac{2\pi}{M} m k} + \sum_{m=1}^{M-1} r_{m}^* e^{j\frac{2\pi}{M} m k} \nonumber \\
& = & \frac{1}{2\Delta} \int_{-\Delta+\theta}^{\Delta+\theta} \left
[ 1 +
2\Re\left \{ \sum_{m=0}^{M-1} e^{-j 2\pi m \omega_k(D,\alpha)} - 1 \right \} \right ] d\alpha \nonumber \\
& = & -1 +  \frac{1}{\Delta} \int_{-\Delta+\theta}^{\Delta+\theta}
\cos\left( \pi \omega_k(D,\alpha)(M-1)\right ) \frac{\sin\left (\pi
\omega_k(D,\alpha)M\right )}{\sin(\pi\omega_k(D,\alpha))} d\alpha,
\end{eqnarray}
where we define the quantity $\omega_k(D,\alpha) = D\sin(\alpha) + k/M$.

In order to obtain the limiting CDF of the eigenvalues of $\Rm$ and find a simple formula
for the asymptotic rank $\rho$,  we obtain an explicit expression of $S(\xi)$ for the autocorrelation
function $r_m = [\Rm]_{\ell,\ell-m}$.
Using (\ref{correlation-matrix-ULA}) and invoking the Lebesgue dominated convergence  theorem,  we have
\begin{eqnarray} \label{super-figo}
S(\xi) & = &
\sum_{m=-\infty}^\infty r_m e^{-j2\pi \xi m} \nonumber \\
& = & \frac{1}{2\Delta} \int_{-\Delta+\theta}^{\Delta+\theta} \left [ \sum_{m=-\infty}^\infty e^{-j2\pi m (D \sin(\alpha) + \xi)} \right ] d\alpha \nonumber \\
& \stackrel{(a)}{=} & \frac{1}{2\Delta} \int_{-\Delta+\theta}^{\Delta+\theta} \left [ \sum_{m=-\infty}^\infty \delta (D \sin(\alpha) + \xi - m) \right ] d\alpha \nonumber \\
& \stackrel{(b)}{=} & \frac{1}{2\Delta} \int_{D
\sin(-\Delta+\theta)}^{D \sin(\Delta+\theta)} \left [
\sum_{m=-\infty}^\infty \delta (z + \xi - m) \right ]
\frac{dz}{\sqrt{D^2 - z^2}},
\end{eqnarray}
where in (a) we used the Poisson sum formula (also known as ``picked
fence miracle'' \cite{lapidoth2009foundation}), in (b) we made the
change of variable $z = D\sin(\alpha)$. The expression
(\ref{super-figo}) is valid for $- \frac{\pi}{2} \leq \theta -
\Delta < \theta + \Delta \leq \frac{\pi}{2}$.
A more general formula, able to recover the classical Bessel $J_0$ autocorrelation function \cite{bello1963characterization}
in the case of uniform isotropic scattering, is provided in Appendix \ref{bessel-derive}.
Owing to the property of the Dirac delta function, we arrive at
\begin{equation}  \label{fikissimo}
S(\xi) = \frac{1}{2\Delta} \sum_{m \in [D\sin(-\Delta+\theta) + \xi, D\sin(\Delta+\theta) + \xi]} \frac{1}{\sqrt{D^2 - (m - \xi)^2}}.
\end{equation}
We have:
\begin{lem} \label{lemma:boundedness}
The function $S(\xi)$ is non-constant over its support and uniformly
bounded, provided that $D \in [0,1/2]$ and $-\phi \leq \theta -
\Delta < \theta + \Delta \leq \phi$ for some constant angle $\phi
\in [0, \pi/2)$.
\end{lem}

\begin{proof}
$S(\xi)$ is periodic and it is sufficient to restrict $\xi$ to the
interval  $[-1/2,1/2]$. As observed before, if $-\frac{\pi}{2} \leq
\theta - \Delta < \theta + \Delta \leq \frac{\pi}{2}$, the general
expression of $S(\xi)$ given in Appendix \ref{bessel-derive}
coincides with (\ref{fikissimo}),
and we have $-D < - D \sin(\phi) \leq D\sin(-\Delta+\theta) < D\sin(\Delta+\theta) \leq D\sin(\phi) < D$.
Since $-1/2 \leq \xi \leq 1/2$ and $D \in [0,1/2]$, the following inequalities hold:
\begin{eqnarray*}
D\sin(-\Delta+\theta) + \xi &\geq& D\sin(-\Delta+\theta) - 1/2 > -D
- 1/2 \geq -1\\
D\sin(\Delta+\theta) + \xi &\leq& D\sin(\Delta+\theta) + 1/2 < D +
1/2 \leq 1.
\end{eqnarray*}
Since $-1 < D\sin(-\Delta+\theta) + \xi <
D\sin(\Delta+\theta) + \xi < 1$, the only integer in the interval $[D\sin(-\Delta+\theta) + \xi, D\sin(\Delta+\theta) +
\xi]$ is 0. Thus,
\begin{equation} \label{ziofanalissimo}
S(\xi) = \frac{1}{2\Delta} \sum_{0 \in [D\sin(-\Delta+\theta) + \xi,
D\sin(\Delta+\theta) + \xi]} \frac{1}{\sqrt{D^2 - \xi^2}}.
\end{equation}
The support $\Sc$ of $S(\xi)$ is the set of values $\xi \in
[-1/2,1/2]$ for which the interval $[D\sin(-\Delta+\theta) + \xi,
D\sin(\Delta+\theta) + \xi]$ contains the point 0, i.e., $\Sc =
[-D\sin(\Delta+\theta), -D\sin(-\Delta+\theta)]$. It is clear by
inspection that $S(\xi)$ is not constant over $\Sc$ (it is
sufficient to observe that $S(\xi)$ is differentiable, and its
derivative is not identically zero over a set of non-zero measure).
To prove that $S(\xi)$ is uniformly bounded, it is sufficient to
notice that the term $\frac{1}{\sqrt{D^2 - \xi^2}}$ in
(\ref{ziofanalissimo}) is real, continuous and finite for all $\xi
\in (-D,D) \supset [-D\sin(\phi), D\sin(\phi)] \supseteq \Sc$.
Hence, it attains its minimum $\kappa'$ and maximum $\kappa_2$ on
$\Sc$, and these are uniformly bounded as~\footnote{We use $\kappa'$
instead of $\kappa_1$ to denote the minimum of $S(\xi)$ on its
support since the minimum eigenvalue, denoted previously by
$\kappa_1$, is generally equal to 0 whenever $\Sc$ is strictly
included in $[-1/2,1/2]$.}
\[ \frac{1}{D} \leq \kappa' < \kappa_2 \leq \frac{1}{D \cos(\phi)} < \infty. \]
\end{proof}

Notice that the assumptions of Lemma \ref{lemma:boundedness} are satisfied for antenna spacing not larger than $\lambda/2$ and
in the assumption, made here, that the ULA receives/transmits energy only in a 120 deg sector (i.e., for AoAs in $[-\pi/3,\pi/3]$).
As a corollary of (\ref{fikissimo}), we obtain the asymptotic rank in closed form:

\begin{thm} \label{asympt-rank}
The asymptotic normalized rank of the channel covariance matrix $\Rm$ with elements defined in (\ref{correlation-matrix-ULA}),
with antenna separation $\lambda D$, AoA $\theta$ and AS $\Delta$, is given by
\begin{equation}
\rho = \min \{ 1, B(D,\theta,\Delta)\},
\end{equation}
where
\begin{equation} \label{interval-size}
B(D,\theta,\Delta) = \left | D\sin(-\Delta+\theta) - D\sin(\Delta+\theta) \right |.
\end{equation}
\end{thm}

\begin{proof}
Notice that $B(D,\theta,\Delta)$ is the size of the interval for $m$ in the summation appearing
in (\ref{fikissimo}). If $B(D,\theta,\Delta) \geq 1$, for any $\xi \in [-1/2,1/2]$ the sum in (\ref{fikissimo})
is non-empty. It follows that $S(\xi) > 0$ for all $\xi$ and the asymptotic normalized rank is $\rho = 1$.
In contrast, if $B(D,\theta,\Delta) < 1$, there exist a set $\Sc^c \subseteq [-1/2,1/2]$ of measure
$1 - B(D,\theta,\Delta)$ for which if $\xi \in \Sc^c$ then the sum in (\ref{fikissimo}) is empty.
Therefore, in this case we have $\rho = B(D,\theta,\Delta)$.
\end{proof}

A good approximation of the actual rank $r$ for large but finite $M$ is given
by $r \approx \rho M$, where $\rho$ is given by  Theorem \ref{asympt-rank}.
Hence, we can predict accurately the rank of the channel covariance from the system geometric
parameters $(D, \theta, \Delta)$.

The empirical CDF of the eigenvalues of $\Rm$ is defined by
\begin{equation} \label{Rspectrum}
F_{\Rm}^{(M)}(\lambda) = \frac{1}{M} \sum_{m=1}^M 1\{ \lambda_m(\Rm) \leq \lambda\}.
\end{equation}
For large $M$, $F_{\Rm}^{(M)}(\lambda)$ can be approximated  either
using (\ref{eigv-circ}) or the collection of samples $\{S([m/M]) : m
= 0,\ldots, M-1\}$, where $[x]$ indicates $x$ modulo the interval
$[-1/2,1/2]$. In both cases, using the resulting collection of $M$
values in (\ref{Rspectrum}), we obtain a convergent approximation
$\widetilde{F}_{\Rm}^{(M)}(\lambda)$ of the empirical CDF
(\ref{Rspectrum}) such that~\cite{grenander1984toeplitz}
\[ \lim_{M \rightarrow \infty} \widetilde{F}_{\Rm}^{(M)}(\lambda) = \lim_{M \rightarrow \infty} F_{\Rm}^{(M)}(\lambda)  = F(\lambda). \]
As an example,  Fig. \ref{cdf1} shows the exact empirical CDF of $\Rm$, its circulant approximation obtained
by (\ref{eigv-circ})) and the asymptotic approximation obtained from the set
$\{S([m/M]) : m = 0,\ldots, M-1\}$,  for a specific choice of the system parameters.
It is apparent that, in this regime, both approximations are very accurate.

\begin{figure}
\centerline{
  \includegraphics[width=10cm]{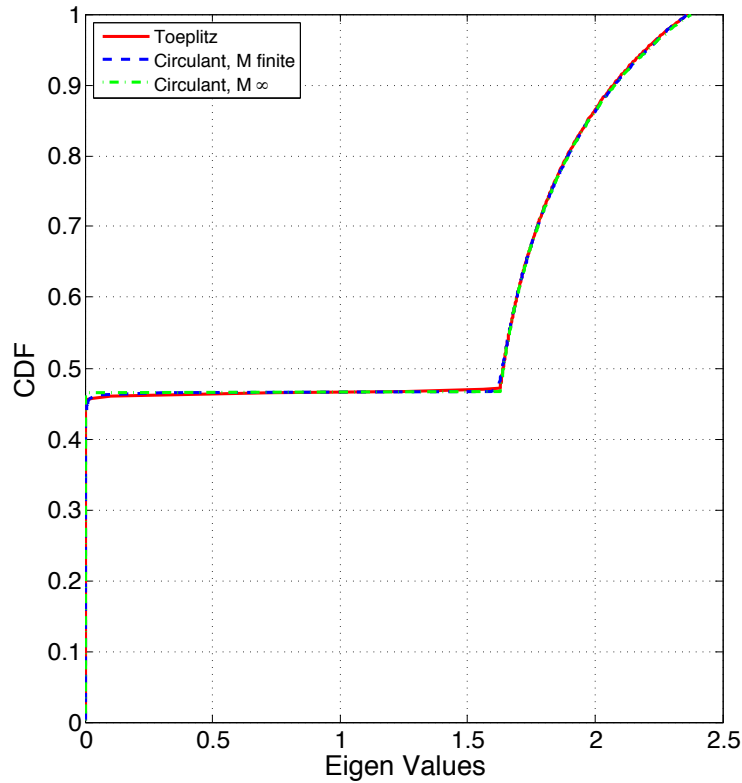}}
  \caption{$M = 400, \theta = \pi/6, D = 1, \Delta = \pi/10$.
  Exact empirical eigenvalue cdf of $\Rm$ (red), its approximation (\ref{eigv-circ}) based on
  the circulant matrix $\Cm$ (dashed blue) and its approximation from the samples of $S(\xi)$ (dashed green).}
  \label{cdf1}
\end{figure}

\subsection{Approximating the channel eigenspace}

Going back to the problem of approximating the eigenvectors of $\Rm$ with a set of DFT columns,
we notice the following properties of $S(\xi)$ in (\ref{fikissimo}):
\begin{enumerate}
\item $S(\xi)$ has support on an interval $\Sc \subseteq [-1/2,1/2]$, of length $\rho$ (see proof of Theorem \ref{asympt-rank}).
\item $S(\xi)$ is non-constant and bounded over its support (see Lemma \ref{lemma:boundedness}).
\end{enumerate}
It follows that $F(\lambda)$ has a single discontinuity at $\lambda= 0$, with jump of height $1 - \rho$, corresponding to the mass-point
of the zero eigenvalues of $\Rm$. For $\rho < 1$, $F(\lambda)$ is continuous over $(0, \kappa_2]$ where $\kappa_2 = \max S(\xi) < \infty$ by Lemma \ref{lemma:boundedness}.
Hence, any interval $[a,b]$ with $0 < a < b \leq \kappa_2$ is a continuity interval of $F(\lambda)$, and the eigenspace
approximation property of Fact \ref{eigenspace-approx} holds. In particular, we have established the following:

\begin{cor} \label{eigenspace-corollary}
Let $\Sc$ denote the support of $S(\xi)$, let $\Jc_{\Sc} = \{ m : [m/M] \in \Sc, m = 0,\ldots, M-1\}$
be the set of indices for which the corresponding ``angular frequency'' $\xi_m = [m/M]$ belongs to $\Sc$,
let $\fv_m$ denote the $m$-th column of the unitary DFT matrix $\Fm$,
%
and let $\Fm_{\Sc} = (\fv_m : m \in \Jc_{\Sc})$ be the DFT submatrix containing
the columns with indices in $\Jc_{\Sc}$.  Then,
\begin{equation}
\lim_{M \rightarrow \infty} \frac{1}{M} \left \| \Um\Um^\herm  - \Fm_{\Sc} \Fm_{\Sc}^\herm \right \|_F^2 = 0,
\end{equation}
where $\Um$ is the $M \times r$ ``tall unitary'' matrix of the non-zero eigenvectors of $\Rm$.
\end{cor}

\begin{proof}
Since $S(\xi)$ is uniformly bounded and strictly positive over $\Sc$, we have
$0 < \min_{\xi \in \Sc}  S(\xi) = \kappa' < \max S(\xi) = \kappa_2$. Hence, letting
$a = \kappa'$ and $b = \kappa_2$,  and using the eigenspace approximation property of Fact \ref{eigenspace-approx} yields the result.
\end{proof}

Consider now a JSDM configuration with an ULA serving $G$ groups
with AoAs within a 120 deg sector. For each group $g$,  we can
approximate the eigenmodes $\Um_g$ by the DFT submatrix
$\Fm_{\Sc_g}$, where $\Sc_g$ denotes the support of $S_g(\xi)$,
given by (\ref{fikissimo}) for AoA $\theta_g$ and AS $\Delta$ (for
simplicity we assume that the AS is common to all groups, although
this can be easily generalized). Corollary
(\ref{eigenspace-corollary}) implies that if $\Sc_g \cap \Sc_{g'} =
\emptyset$ (disjoint angular frequency support), then
$\Fm_{\Sc_g}^\herm \Fm_{\Sc_g'} = \zerov$. It follows that if the
$G$ groups are chosen to have spectra with disjoint support,  then
$[\Fm_{\Sc_1}, \ldots, \Fm_{\Sc_G}]$ is {\em exactly} tall unitary
and, because of Fact \ref{eigenspace-approx}, $\underline{\Um} =
[\Um_1, \ldots, \Um_G]$ is approximately tall unitary, for large
$M$. The following result provides such condition expressed directly
in terms of the AoA intervals.

\begin{thm} \label{nonoverlapping}
Groups $g$ and $g'$ with angle of arrival $\theta_g$ and $\theta_{g'}$
and common angular spread $\Delta$ have spectra with disjoint support if their AoA intervals
$[\theta_g - \Delta,\theta_g + \Delta]$ and $[\theta_{g'} - \Delta,\theta_{g'} + \Delta]$ are disjoint.
\end{thm}

\begin{proof}
Define:
\begin{eqnarray*}
A_g &=& \max(D \sin (\theta_g + \Delta), D \sin(\theta_g - \Delta))\\
B_g &=& \min(D \sin (\theta_g + \Delta), D \sin(\theta_g - \Delta))\\
A_{g'} &=& \max(D \sin (\theta_{g'} + \Delta), D \sin(\theta_{g'} - \Delta))\\
B_{g'} &=& \min(D \sin (\theta_{g'} + \Delta), D \sin(\theta_{g'} - \Delta)).
\end{eqnarray*}
From (\ref{fikissimo}) we notice that $S_g(\xi)$ and $S_{g'}(\xi)$
have disjoint supports if $A_g \leq B_{g'}$ or $A_{g'} \leq B_g$.
Since the mapping $x \mapsto \sin(x)$ is one-to-one in the interval $[-\pi/3,\pi/3]$, this condition corresponds to
$[\theta_g - \Delta,\theta_g + \Delta] \cap [\theta_{g'} - \Delta,\theta_{g'} + \Delta] = \emptyset$.
\end{proof}

\subsection{DFT pre-beamforming}

\begin{figure}
\centering
  \includegraphics[width=10cm]{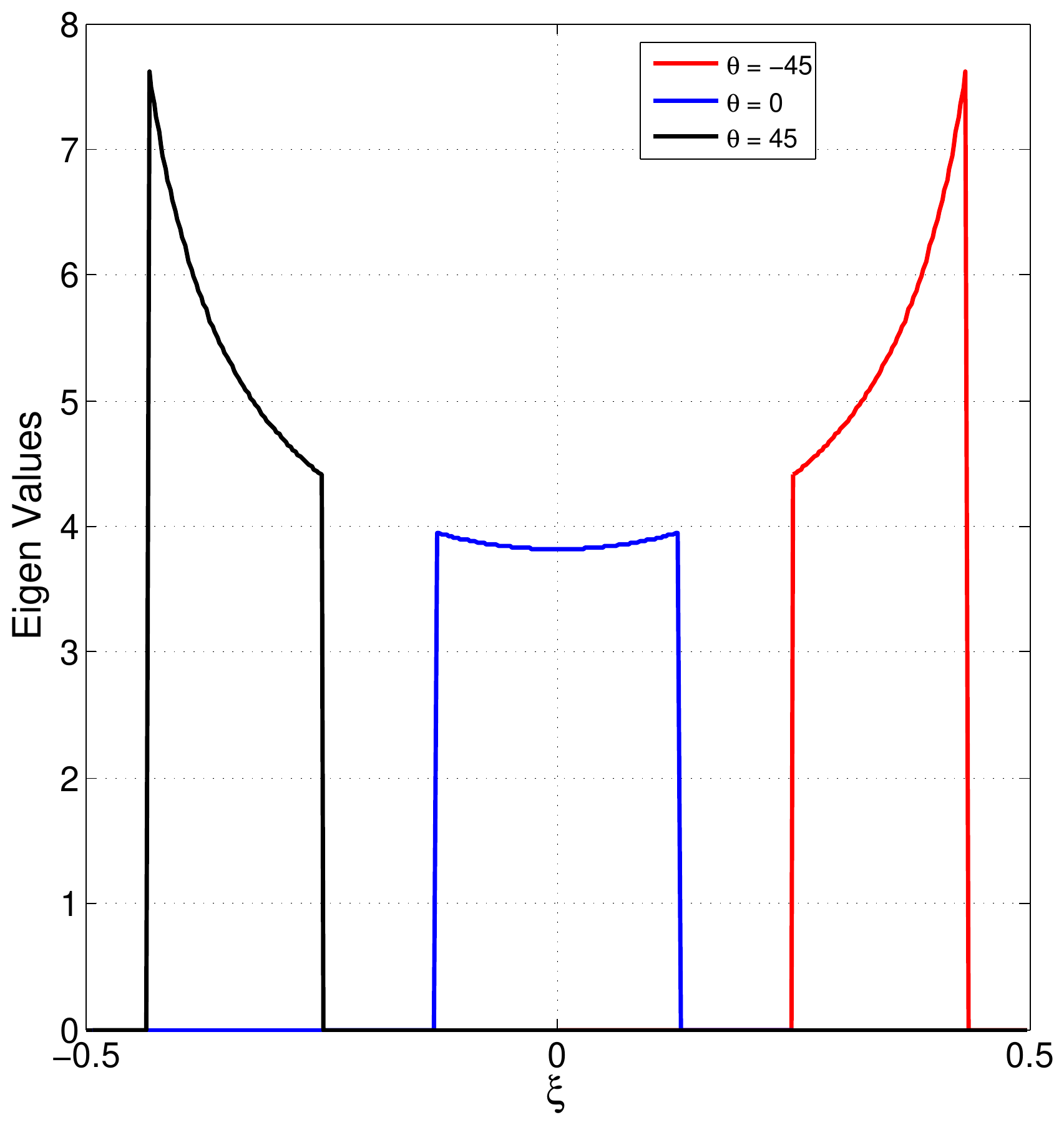}
  \caption{Eigenvalue spectra for a ULA with $M = 400$, $G = 3$, $\theta_1 = \frac{-\pi}{4}, \theta_2 = 0, \theta_3 = \frac{\pi}{4}$,
  $D = 1/2$ and $\Delta = 15$ deg.}
  \label{fig:eig-val-spectra}
\end{figure}

\begin{figure}
\centering
\includegraphics[width=10cm]{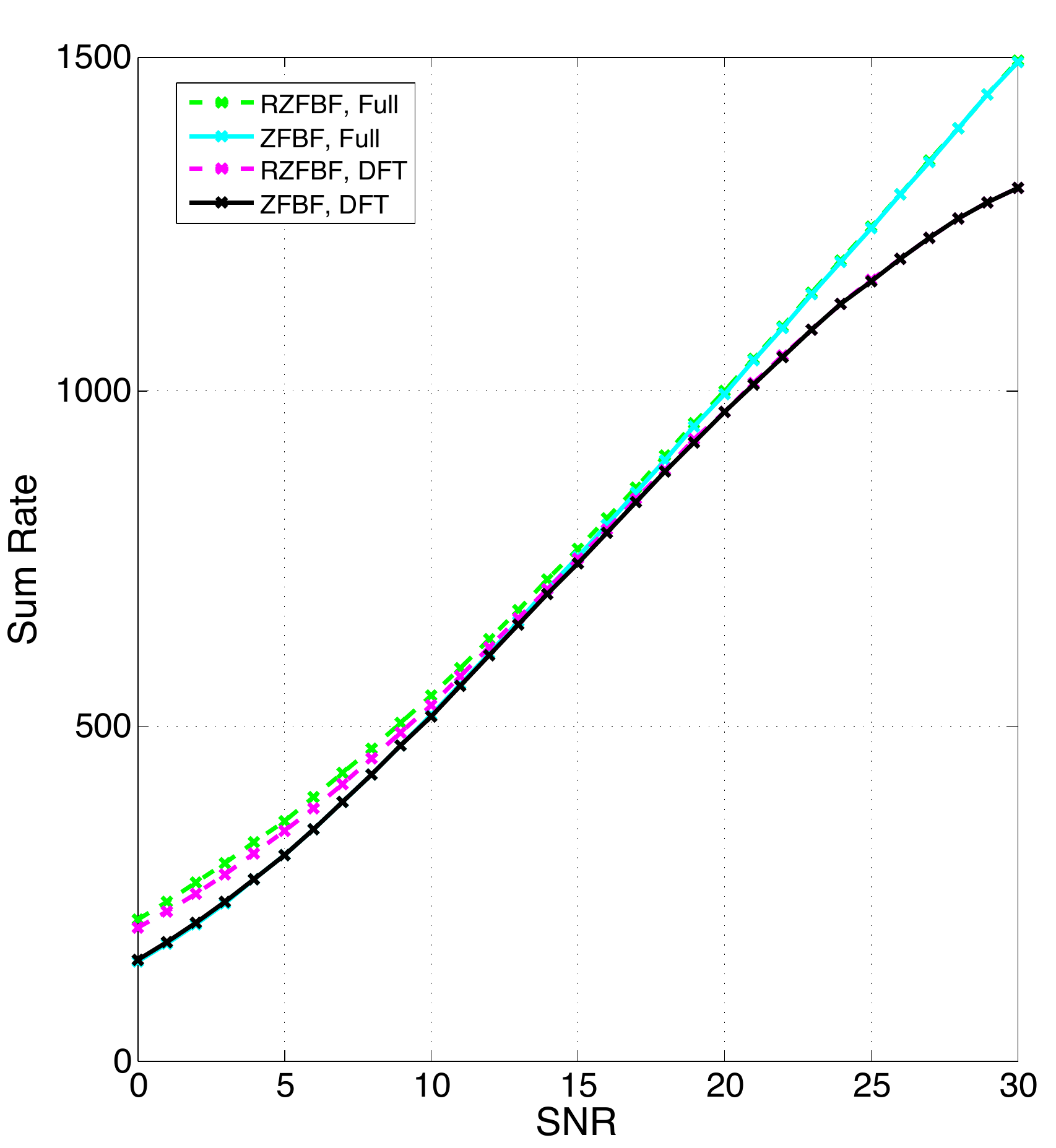}
 \caption{Sum spectral efficiency (bit/s/Hz) vs. SNR (dB) for JSDM (computed via deterministic equivalents)
  for DFT pre-beamforming and PGP, for the configuration with spectra shown in Fig.~\ref{fig:eig-val-spectra},
  choosing $b_g = r_g$ for all groups $g = 1,2,3$.}
  \label{fig:dft-results}
\end{figure}

Owing to the asymptotic eigenspace approximation and mutual orthogonality of the previous section, an efficient approach to JSDM design when the
BS is equipped with a large ULA per sector consists of selecting groups of users with (almost) identical AoA intervals, and find $G$ groups of such users
with non-overlapping AoA intervals. Then, we let $\Bm_g = \Fm_{\Sc_g}$, for $g = 1,\ldots, G$, with $\Fm_{\Sc_g}$ defined as in Corollary \ref{eigenspace-corollary}.
It follows that $\Fm_{\Sc_g}^\herm \Um_{g'} \approx \zerov$ for all $g \neq g'$, such that the sum spectral efficiency achieved by JSDM with PGP
is close to the sum spectral efficiency of the corresponding MU-MIMO downlink channel with full CSIT (see Theorem \ref{simple-opt}).
Notice that this approach is particularly attractive since only a coarse parametric knowledge (AoA interval) for each user is required, rather than an
accurate estimate of its channel covariance matrix.

Fig.~\ref{fig:eig-val-spectra} show the spectra $S_g(\xi)$  for $g = 1,2,3$,  $M
= 400$, and $\theta_1 = \frac{-\pi}{4}, \theta_2 = 0, \theta_3 =
\frac{\pi}{4}$, with $D = 1/2$ and $\Delta = 15$ deg.  The
performance of JSDM with PGP and DFT pre-beamforming is shown in Fig.~\ref{fig:dft-results}, indicating that up to 20 dB of SNR, DFT
pre-beamforming performs close to schemes with full CSIT.

\section{JSDM with 3D pre-beamforming}   \label{sec:super-massive}

So far we considered a planar geometry where each user group $g$ is
identified by its  AoA interval $[\theta_g - \Delta, \theta_g +
\Delta]$. For the sake of simplicity, we allocated equal power to
all $S$ downlink data streams. This is a near-optimal power
allocation in the high SNR (high spectral efficiency) regime and in
the case where the pathloss from the BS to all the UTs is
approximately equal. In practice, however, users with same (or very
similar) AoA interval may be located at different distances to the
BS. In this case, a simple alternative to the complicated and
generally non-convex power allocation optimization across different
users~\footnote{While for MU-MIMO with full CSIT and optimal
capacity achieving coding \cite{weingarten2006capacity} the power
allocation is  a convex optimization problem that can be efficiently
solved \cite{Wei06}, for JSDM with either JGP and PGP, the problem
is non-convex and the optimization is not amenable to a
computationally efficient solution.} consists of dividing the cell
into concentric annular regions, and serve simultaneously groups in
the same region, such that the pathloss is nearly equal for all
jointly processed groups. Groups in different annular regions can be
scheduled over the time-frequency slots. In this section, we
consider an extension of this approach where we assume that the BS
is elevated with respect to ground.   For example, antenna elements
could be placed on the window frames of a tall building forming a
rectangular array with $M$ antennas in each row (each row is an ULA)
and a total of $N$ rows in the vertical dimension. By exploiting the
vertical dimension, different annular regions can be served
simultaneously in the spatial domain.


Assuming a rectangular $N \times M$ array,  we consider using a {\em separable} 3D
pre-beamforming scheme: beamforming in the elevation angle dimension is used to form beams that
``look down'' at different angles, i.e., they illuminate concentric annular regions within the cell sector.
For each such region, precoding in the azimuth angle dimension is obtained by JSDM scheme with $M$ antennas,
as done before. Thanks to separability, we can optimize JSDM schemes independently, one for each
annular region.

The groups served simultaneously by JSDM in the same region
are now identified by two indices, $(l,g)$ where $l = 1,\ldots,L$
indicates the annular region and $g = 1, \ldots, G_l$ the group in
each $l$-th region. A set of groups served simultaneously, on the same time-frequency dimensions, is referred to as a ``pattern''.
A pattern does not necessarily cover the whole sector. In fact, it is usually better to allow for ``holes''
in the pattern, i.e., the group footprints  can be separated by gaps,
in order to guarantee near orthogonality between the dominant eigenmodes of the groups in the same pattern
and thus limiting inter-group interference with PGP.  In
order to provide coverage to the whole sector, different {\em intertwined patterns} can be multiplexed over the
time-frequency dimension, similarly to the intertwined cooperative pattern idea proposed in
\cite{ramprashad2009cellular},\cite{ramprashad2010joint},\cite{caire2010rethinking}.
The fraction of the time-frequency dimensions allocated to each
pattern can be further optimized in order to maximize a network
utility function, reflecting  some desired notion of fairness (see for example \cite{Huh11}).

For the time being, we focus on a single pattern comprising $L$
regions in the elevation angle dimension, and $G_l$ groups in the
azimuth angle dimension for each region $l = 1,\ldots, L$. We let $K_{l,g}$ denote the number of users
in group $(l,g)$. At the BS, an $N \times M$ rectangular
antenna array with $N$ rows and $M$ columns is used.
For each region $l$, we denote by $\Rm_{V,l} \in \CC^{N \times N}$ the vertical
channel covariance matrix~\footnote{We assume that the vertical
correlation does not depend on $g$, but just on $l$.}  and, for each
group $(l,g)$, we let $\Rm_{H,l,g} \in \CC^{M \times M}$ denote the
the horizontal channel covariance matrix. $\Rm_{V,l}$ and
$\Rm_{H,l,g}$ are modeled according to (\ref{eq:SM-4}), with the eigen-decompositions:
\begin{equation}
\Rm_{V,l} = \Um_{V,l} \Lambdam_{V,l} \Um_{V,l}^\herm, \;\;\; \mbox{and} \;\;\;
\Rm_{H,l,g} = \Um_{H,l,g} \Lambdam_{H,l,g} \Um_{H,l,g}^\herm.
\end{equation}
Letting $\hv_{l,g_k}$ denote the $MN \times 1$ the vectorized
channel from the $M \times N$ BS array to the $k^{\rm th}$ user in
group $(l,g)$, we have
\begin{equation}
\EE[\hv_{l,g_k} \hv_{l,g_k}^\herm] = \Rm_{l,g} = \Rm_{H,l,g} \otimes
\Rm_{V,l} = (\Um_{H,l,g} \otimes \Um_{V,l})(\Lambdam_{H,l,g} \otimes
\Lambdam_{V,l})(\Um_{H,l,g}^\herm \otimes \Um_{V,l}^\herm).
\end{equation}
This covariance matrix is common (by assumption) to all users
$g_k$ in group $(l,g)$. Denoting the ranks of $\Rm_{H,l,g}$ and $\Rm_{V,l}$ by
$r_{H,l,g}$ and $r_{V,l}$, respectively, we write $\hv_{l,g_k}$ as
\[ \hv_{l,g_k} = (\Um_{H,l,g} \otimes \Um_{V,l})(\Lambdam_{H,l,g}^{\frac{1}{2}} \otimes
\Lambdam_{V,l}^{\frac{1}{2}}) \wv_{l,g_k}, \]
where $\Um_{H,l,g}$ is $M \times r_{H,l,g}$, $\Um_{V,l}$ is $N \times r_{V,l}$,
$\Lambdam_{H,l,g}$ is $r_{H,l,g} \times r_{H,l,g}$ and
$\Lambdam_{V,l}$ is $r_{V,l} \times r_{V,l}$. The vector
$\wv_{l,g_k}$, of dimension $r_{H,l,g}r_{V,l} \times 1$, has i.i.d.  entries $\sim \Cc\Nc(0,1)$.

In JSDM with 3D pre-beamforming, the transmitted signal is given by
\begin{equation}
\xv = \sum_{l=1}^L (\Bm_l \Pm_l \dv_l) \otimes \qv_l,
\end{equation}
where $\qv_l$ denotes the $N \times 1$ pre-beamforming vector for
region $l$ in the elevation angle dimension, $\Bm_l$ is the $M
\times b_l$ pre-beamforming matrix of the form $\Bm_l =
[\Bm_{l,1},\ldots,\Bm_{l,G_l}]$, where $\Bm_{l,g}$ denotes the
pre-beamforming matrix of size $M \times b_{l,g}$ for group $(l,g)$
and $\Pm_l$ is the linear precoding matrix for the groups of region
$l$, that depends on the instantaneous effective channels as given
in Section \ref{sec:jsdm}. Notice that we allocate (by design) a
single dimension per region in the elevation angle direction (this
is why $\qv_l$ has dimensions $N \times 1$) since, because of the
relatively small angle under which the BS sees the different
regions, it is realistic to expect that $\Rm_{V,l}$ has a single
dominant eigenmode. Generalizations considering higher dimensional
vertical pre-beamforming for each region are conceptually
straightforward, although not very useful in typical practical
scenarios.

Using repeatedly the Kronecker product rule $(\Am \otimes \Bm) (\Cm \otimes \Dm) = (\Am\Cm) \otimes (\Bm \Dm)$,
the received signal for user $g_k$ in group $(l,g)$ can be written as
\begin{eqnarray}
\label{eq:1} y_{l,g_k} &=& \wv_{l,g_k}^\herm
(\Lambdam_{H,l,g}^{\frac{1}{2}} \otimes
\Lambdam_{V,l}^{\frac{1}{2}}) (\Um_{H,l,g}^\herm \otimes
\Um_{V,l}^\herm) \xv + z_{l,g_k}\nonumber\\
&=&  \wv_{l,g_k}^\herm (\Lambdam_{H,l,g}^{\frac{1}{2}} \otimes
\Lambdam_{V,l}^{\frac{1}{2}}) (\Um_{H,l,g}^\herm \otimes
\Um_{V,l}^\herm) \left[\sum_{m=1}^L (\Bm_m \Pm_m \dv_m) \otimes
\qv_m \right] + z_{l,g_k}\nonumber\\
&=& \wv_{l,g_k}^\herm (\Lambdam_{H,l,g}^{\frac{1}{2}} \otimes
\Lambdam_{V,l}^{\frac{1}{2}}) \sum_{m=1}^L\left[ (\Um_{H,l,g}^\herm
\Bm_m \Pm_m \dv_m) \otimes (\Um_{V,l}^\herm \qv_m) \right] + z_{l,g_k}\nonumber\\
&=& \wv_{l,g_k}^\herm (\Lambdam_{H,l,g}^{\frac{1}{2}} \otimes
\Lambdam_{V,l}^{\frac{1}{2}}) \sum_{m=1}^L\left[ (\Um_{H,l,g}^\herm
\Bm_m) \otimes (\Um_{V,l}^\herm \qv_m)
\right] \; \Pm_m \dv_m  + z_{l,g_k}.
\end{eqnarray}
If $\qv_m$ is chosen to be orthogonal to $\mbox{Span}(\{ \Um_{V,l} : l \neq m \})$,  (\ref{eq:1}) reduces to
\begin{equation}
\label{eq:2} y_{l,g_k} = \wv_{l,g_k}^\herm
(\Lambdam_{H,l,g}^{\frac{1}{2}} \otimes
\Lambdam_{V,l}^{\frac{1}{2}}) \left[ (\Um_{H,l,g}^\herm \Bm_l)
\otimes (\Um_{V,l}^\herm \qv_l) \right]  \Pm_l \dv_l  + z_{l,g_k}.
\end{equation}
Stacking the signals $y_{l,g_k}$ for all users $g_k$ in group
$(l,g)$ into a $K_{l,g} \times 1$ vector $\yv_{l,g}$, we obtain
\begin{equation}  \label{super-massive-group}
\yv_{l,g} = \Wm_{l,g}^\herm
(\Lambdam_{H,l,g}^{\frac{1}{2}} \otimes
\Lambdam_{V,l}^{\frac{1}{2}}) \left[ (\Um_{H,l,g}^\herm \Bm_l)
\otimes (\Um_{V,l}^\herm \qv_l) \right] \Pm_l \dv_l  + \zv_{l,g},
\end{equation}
where we let $\Wm_{l,g} = [\wv_{l,g_1}, \ldots,
\wv_{l,g_{K_{l,g}}}]$ and $\zv_{l,g} = [z_{l,g_1}, \ldots,
z_{l,g_{K_{l,g}}}]^\transp$.


If the regions are sufficiently separated in the elevation angle
dimension, it is possible to align $\qv_l$ with the dominant
eigenmode of $\Um_{V,l}$, while maintaining the orthogonality
condition $\Um_{V,m}^\herm \qv_l = 0$ for $m \neq l$. In this case,
we have $\Um_{V,l}^\herm \qv_l = (1,0,\ldots,0)^\transp$ and
(\ref{super-massive-group}) reduces to the same form treated
previously for the planar geometry,  with an additional
region-specific coefficient $\sqrt{\lambda_{V,1}}$, corresponding to
the largest eigenvalue of the matrix $\Lambdam_{V,l}$:
\begin{equation} \label{super-massive-group1}
\yv_{l,g} = \sqrt{\lambda_{V,1}}\Wm^\herm_{l,g}
\Lambdam_{H,l,g}^{\frac{1}{2}} \Um_{H,l,g}^\herm \Bm_l \Pm_l \dv_l + \zv_{l,g}.
\end{equation}
Stacking the vectors $\yv_{l,g}$ for all $g = 1,\ldots, G_l$, we obtain
\begin{equation}
\label{eq:jsdm-super-mimo} \yv_{l} = \sqrt{\lambda_{V,1}} \left[
\begin{array}{c}
\Wm^\herm_{l,1}\Lambdam_{H,l,1}^{\frac{1}{2}} \Um_{H,l,1}^\herm\\
\Wm^\herm_{l,2}\Lambdam_{H,l,2}^{\frac{1}{2}} \Um_{H,l,2}^\herm\\
\vdots\\
\Wm^\herm_{l,G_l}\Lambdam_{H,l,G_l}^{\frac{1}{2}} \Um_{H,l,G_l}^\herm\\
 \end{array}\right] \Bm_l \Pm_l \dv_l +
\zv_{l},
\end{equation}
which is of the same form as (\ref{eq:SM-4b}). At this point, the
pre-beamforming matrix $\Bm_l$ and the precoding matrix $\Pm_l$ can
be optimized independently for each region $l$, as described before
for the planar geometry. The coefficients $\lambda_{V,1}$
incorporate the effect of the different geometry of the annular
regions in the elevation angle dimension, including the path loss
due to different distances of the regions from the BS. The
allocation of the total transmit power over the regions can be
further optimized.

\subsection{Results with 3D pre-beamforming}

\begin{figure}
\centerline{\includegraphics[width=6cm]{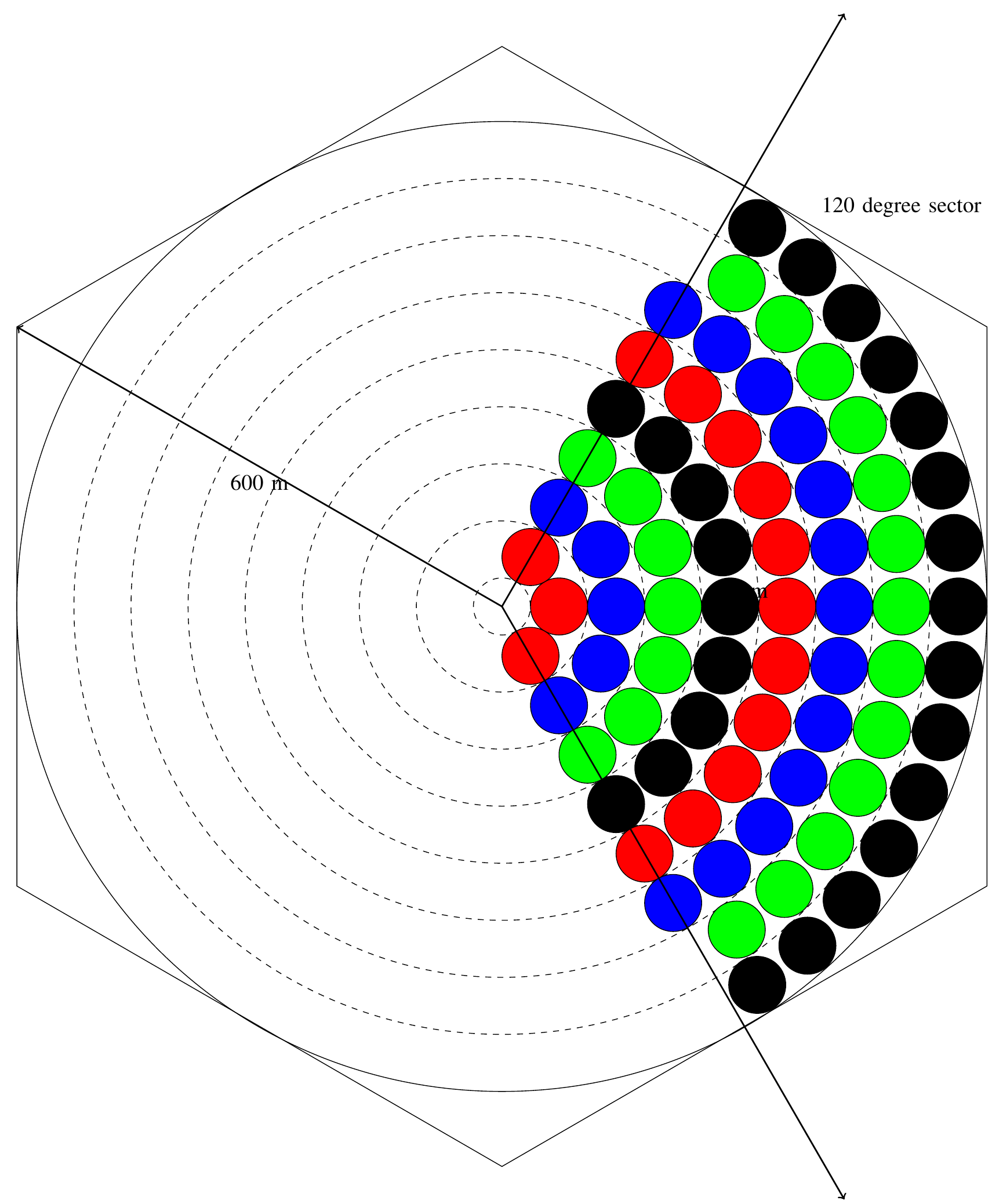}}
  \caption{The layout of one pattern for JSDM with 3D pre-beamforming. The concentric regions are separated by the vertical pre-beamforming.
  The circles indicate user groups. Same-color groups are served simultaneously using JSDM.}
\label{super-mimo-cellular-fig}
\end{figure}

\begin{figure}
\centerline{\includegraphics[width=10cm]{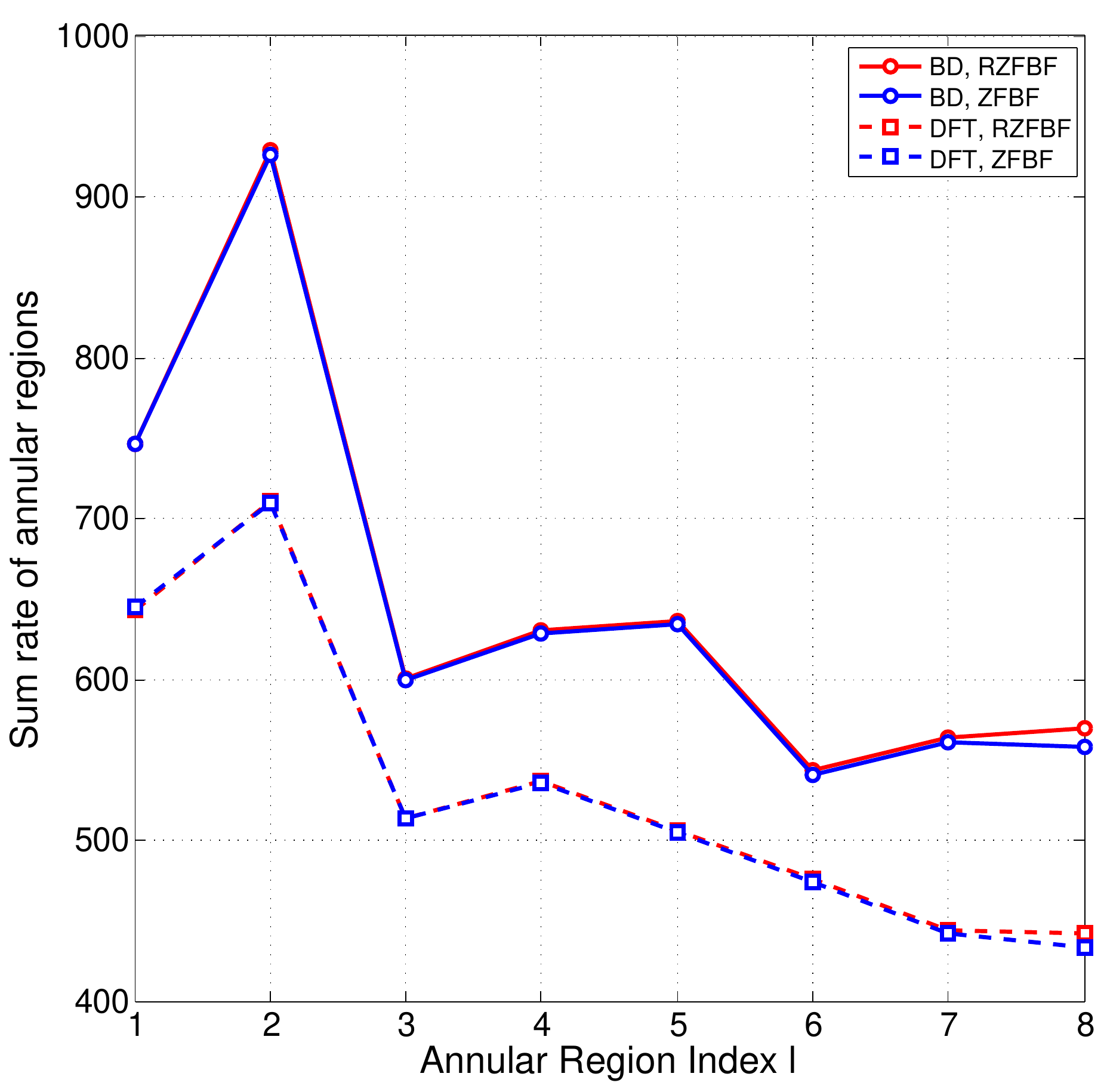}}
  \caption{Sum spectral efficiency $\bar{\Rc}_l$ for different annular regions $l = 1,\ldots,8$ with Regularized ZF and ZF for JSDM with 3D pre-beamforming and ideal CSIT.
  ``BD'' denotes PGP with approximate block diagonalization and ``DFT'' stands for PGP with DFT pre-beamforming.
  Equal power is allocated to all served users and the number of users (streams) in each group is optimized in order to maximize the overall spectral efficiency.}
\label{cellular-sim-fig-1}
\end{figure}

We present some results for JSDM with 3D pre-beamforming and PGP,
with either BD or DFT pre-beamforming in each region. The system
layout is shown in Fig.~\ref{super-mimo-cellular-fig}. We consider
one sector of a hexagonal cell of radius 600 m. The scattering rings
in the channel correlation model have radius ${\sf r}=30$ m. The BS
is located at the center of the cell with the antennas at an
elevation of ${\sf h} = 50$ m, and is equipped with a rectangular
array with $M = 200$ and $N = 300$. We partition the sector into 8
concentric regions at distance $60l$ m, $l \in \{1,\ldots,8\}$. Each
annular region is divided into small scattering rings, each defining
a group. The pathloss between the BS and a point at distance ${\sf
x}$ m is given by
\begin{equation}
g({\sf x}) = \frac{1}{1 + (\frac{{\sf x}}{{\sf d}_0})^\delta},
\end{equation}
with $\delta = 3.8$, ${\sf d}_0 = 30$ m. In these results we assume
ideal CSIT for computing the JSDM precoder. The horizontal
covariance matrix for all groups $(l,g)$ is given by (\ref{eq:SM-4})
with $\Delta_{H,l} = \arctan(\frac{r}{60l}) = \arctan(\frac{1}{2l})$
and $\theta_{H,l,g} \in [-\pi/3,\pi/3]$ such that for any two groups
$(l,g_1)$ and $(l,g_2)$, we have $|\theta_{H,l,g_1} -
\theta_{H,l,g_2}| > 2 \Delta_l$. It is easy to see from
Fig.~\ref{super-mimo-cellular-fig} that as the distance of the
concentric regions from the BS increases, more and more user groups
can be accommodated in the annular region, since $\Delta_{H,l}$
decreases. The vertical covariance matrix is again given by
(\ref{eq:SM-4}) with $\Delta_{V,l} = \frac{1}{2}(\arctan(\frac{60l +
{\sf r}}{{\sf h}}) - \arctan(\frac{60l - {\sf r}}{{\sf h}}))$ and
$\theta_{V,l} = \frac{1}{2}(\arctan(\frac{60l + {\sf r}}{{\sf h}})
+\arctan(\frac{60l-{\sf r}}{{\sf h}}))$. Since the total angle under
which the sector is seen from the elevation viewpoint is narrow, a
large number of antennas in the vertical direction is required in
order to achieve orthogonality between all annular regions
eigenmodes.

For finite $N$, in order to guarantee a desired angular separation between annular regions and therefore
have near-orthogonality in the elevation angle dimension, it is convenient
to partition the annular regions into maximally separated subsets (patterns) and apply BD in the
vertical pre-beamforming. Different patterns can be scheduled in different time-frequency slots.
We denote by $\Ac = \{\Ac_1,\Ac_2,\ldots\}$ the set of patterns. Finding the best possible pattern
partition is computationally hard, so for the sake of simplicity, we
consider a simple partitioning as shown in Fig.~\ref{super-mimo-cellular-fig}, where annular regions with the same
color belong to the same pattern. In our example (see Fig.~\ref{super-mimo-cellular-fig}), numbering the annular regions in ascending
order based on their proximity to the BS, we have $\Ac = \{ \{1,5\} , \{2,6\} , \{3,7\} , \{4,8\} \}$. In this way, (\ref{eq:jsdm-super-mimo}) is replaced by
\begin{equation}
\label{eq:jsdm-super-mimo-2} \yv_{l} = |
\Lambdam_{V,l}^{\frac{1}{2}}\Um_{V,l}^\herm \qv_l | \left[
\begin{array}{c}
\Wm^\herm_{l,1}\Lambdam_{H,l,1}^{\frac{1}{2}} \Um_{H,l,1}^\herm\\
\Wm^\herm_{l,2}\Lambdam_{H,l,2}^{\frac{1}{2}} \Um_{H,l,2}^\herm\\
\vdots\\
\Wm^\herm_{l,G_l}\Lambdam_{H,l,G_l}^{\frac{1}{2}} \Um_{H,l,G_l}^\herm\\
 \end{array}\right] \Bm_l \Pm_l \dv_l +
\zv_{l}.
\end{equation}
Notice that due to BD in the vertical direction, the inter-region interference is exactly zero since $\Um_{V,m}^\herm \qv_l  = \zerov$ for $m \neq l$.
Within each annular region, we use JSDM with PGP. The pre-beamforming matrices $\Bm_l$ for region $l$ are obtained using
approximate BD or the DFT method, as discussed in previous sections.
The dominant rank $r^\star_{l,g}$ for each group $(l,g)$  is given by
\begin{equation}
r^\star_{l,g} = MD(\sin(\theta_{H,l,g} + \Delta_{H,l}) -
\sin(\theta_{H,l,g} - \Delta_{H,l})),
\end{equation}
which is a good approximation for large $M$ motivated by Theorem \ref{asympt-rank}.
For simplicity, we do not consider noisy CSIT and assume that the BS has full knowledge of
the effective channels.  Hence, we let $b_{l,g} = r^\star_{l,g}$. In contrast,
in the case of noisy CSIT the parameter $b_{l,g}$ should be optimized for given channel coherence
block length $T$, as discussed in Remark \ref{choice-of-bprime-remark}.

Denoting by $\Rc_q$ the sum spectral efficiency of pattern $\Ac_q$, and letting $Q$ denote the number of patterns,
the network utility maximization problem is given by
\begin{eqnarray}
\max && g(\Rc_1,\ldots,\Rc_Q)\nonumber\\
\mbox{subject to} && \Rc_q \leq \nu_q \Rc_q^*,\ \ \mbox{for} \; q = 1,\ldots, Q, \nonumber\\
&& \sum_{q=1}^Q \nu_q = 1.
\end{eqnarray}
where $g(\cdot)$ is a concave component-wise non-decreasing network utility function capturing some desired
notion of {\em fairness}, and the optimization variables $\{\nu_q\}$ are the fractions of time-frequency dimensions allocated to each pattern.
We define
\[ \Rc_q^* = \sum_{l \in \Ac_q} \sum_{g = 1}^{G_l} \sum_{k = 1}^{S_{l,g}} R_{l,g,k}, \]
to be the spectral efficiency of each individual pattern, where $S_{l,g}$ is the number of downlink streams to group $(l,g)$ and
$R_{l,g,k}$ is the rate of the $k$-th stream of group $(l,g)$.
We have considered two cases of fairness: proportional fairness (PFS), and max-min fairness.
In both cases, the optimal dimension allocation fractions $\{\nu_q\}$ can be found in closed form.
For PFS, we have  $g(\Rc_1,\ldots,\Rc_Q) = \sum_{q=1}^Q \log (\Rc_q)$, yielding the solution $\nu_q = \frac{1}{Q}$ for all $q$.
For max-min fairness, we have $g(\Rc_1,\ldots,\Rc_Q) = \min_q \Rc_q$, yielding the solution
$\nu_q = \frac{\frac{1}{\Rc^*_q}}{\sum_q \frac{1}{\Rc^*_q}}$.

The spectral efficiency $\Rc_q^*$ can be optimized independently for
each pattern $\Ac_q$. For a given JSDM precoding  scheme, we need to
search over the number of downlink streams in each group. This is a
multi-dimensional integer search over the parameters $\{S_{l,g}\}$
for all groups $(l,g) \in \Ac_q$. In addition, we should optimize
with respect to the power allocation to the downlink data streams,
as mentioned before. In order to obtain a tractable problem, we
resort to good heuristics. Following the design guideline given in
Remark \ref{choice-of-Sprime-remark}, we know that the ratio
$S_{l,g}/b_{l,g}$ should be approximately the same for the optimal
$S_{l,g}$ for all groups $(l,g)$ with similar geometry, i.e.,
belonging to the same region. Hence, we fix this ratio to be the
same for all groups in the same region, and indicate it as
$\alpha_l$. In addition, as done before, we restrict to equal power
allocation to all the downlink streams. Indicating this common
per-stream power value by $\bar{P}$, and letting
$R_{l,g,k}(\bar{P})$ denote the rate of the $k$-th stream of group
$(l,g)$ as a function of $\bar{P}$, calculated according to the
methods given in Section \ref{sec:PERF} and Appendix
\ref{sec:determ-equiv-nonideal-csi}, for given MU-MIMO precoding
scheme, the optimization with respect to $\{\alpha_l\}$ is expressed
by
\begin{eqnarray} \label{optz-prob-3}
\max & & \sum_{l \in \Ac_q} \sum_{g = 1}^{G_l} \sum_{k = 1}^{S_{l,g}} R_{l,g,k}(\bar{P}) \nonumber\\
\mbox{subject to } & & S_{l,g} = \lfloor \alpha_l b_{l,g} \rfloor\nonumber\\
&& \bar{P} = \frac{P}{\sum_{l \in \Ac_q} \sum_{g=1}^{G_l} S_{l,g}}.
\end{eqnarray}
Notice that for a pattern with $|\Ac_q|$ regions, (\ref{optz-prob-3}) consists of a $|\Ac_q|$-dimensional search over the real parameters $\alpha_l \in [0,1]$, which
is tractable when $|\Ac_q|$ is small (in our case, $|\Ac_q| = 2$).

Fig.~\ref{cellular-sim-fig-1} shows the sum spectral efficiency $\sum_{g = 1}^{G_l} \sum_{k = 1}^{S_{l,g}} R_{l,g,k}(\bar{P})$
for each annular region $l = 1,\ldots,8$ in the setup of Fig.~\ref{super-mimo-cellular-fig} with system parameters given at the beginning of this section,
resulting from the above optimization for both DFT pre-beamforming and approximate BD using PGP with RZFBF and ZFBF precoding.
The corresponding sum spectral efficiencies under PFS and max-min fairness scheduling are reported in
Table \ref{tab:cellular-rates}.

\begin{table}
\centering \caption{Sum spectral efficiency (bit/s/Hz) under PFS and max-min fairness
scheduling for PGP and approximate BD/DFT.}
\begin{tabular}{|c|c|c|}
    \hline Scheme & Approximate BD & DFT based\\
    \hline PFS, RZFBF &  1304.4611 & 1067.9604\\
    \hline PFS, ZFBF & 1298.7944 & 1064.2678\\
    \hline MAXMIN, RZFBF & 1273.7203 & 1042.1833\\
    \hline MAXMIN, ZFBF & 1267.2368 & 1037.2915\\
    \hline
\end{tabular}
\label{tab:cellular-rates}
\end{table}

\section{Concluding remarks}

In this work we proposed Joint Space-Division and Multiplexing
(JSDM), a novel approach to MU-MIMO downlink that requires reduced
channel estimation downlink training overhead and CSIT feedback and
therefore is potentially suited to FDD systems, despite using a
large number of BS antennas. JSDM exploits the fact that for large
BSs, mounted on the top of a building or on a dedicated tower,
channel vectors are far from isotropically distributed. Instead,
their dominant eigenspace has dimension much smaller than the number
of BS antennas. Different groups of users are selected, such that
the users in each group share (approximately) the same dominant
eigenspace,  and  the eigenspaces of different groups are nearly
orthogonal. JSDM serves simultaneously such groups of users, and
multiple users in each group. The separation of the groups in the
spatial domain (space-division) is obtained through a
pre-beamforming matrix that depends only on the channel covariance
matrices, while the multiplexing of multiple users in each group is
obtained via linear MU-MIMO precoding based on the instantaneous
``effective'' channel, including the pre-beamforming. It turns out
that the effective channel has reduced dimensionality with respect
to the original multi-antenna multiuser channel, especially with a
``per-group processing'' (PGP) approach, i.e., where each group is
individually pre-coded, disregarding the inter-group interference.
JSDM with PGP can be regarded as a generalization of sectorization,
where each group acts as a directional sector, and in each sector we
apply MU-MIMO spatial multiplexing, disregarding inter-sector
interference.

We showed that when the collection of the channel covariance eigenvectors of the groups forms a tall unitary matrix,
then JSDM with PGP is optimal, in the sense that it can achieve the capacity of the underlying MU-MIMO channel with full instantaneous CSIT.
Then, using Szego's asymptotic theory of large Toeplitz random matrices, we showed that when the BS is equipped with a large linear uniform array,
this tall unitary condition is closely approached, and the pre-beamforming matrix can be obtained by selecting an appropriate subset of columns of
a unitary DFT matrix. In fact, under these assumptions the accurate estimation of the channel covariance matrix is not needed, and just a coarse
estimation of the AoA range for each group is sufficient, as long as the AoA ranges of different groups do not overlap in the azimuth angle domain.
Finally, we extended our approach to the case of 3D beamforming, considering rectangular arrays and pre-beamforming in the elevation angle (vertical) direction.
In this case, the proposed JSDM scheme partitions the cell into concentric annular regions, and serves groups of users with different azimuth angle
in each region. We demonstrated the effectiveness of the proposed scheme in the case of a typical cell size, typical propagation pathloss, and
a large rectangular antenna array mounted on the face of a tall building. In our case, under ideal CSIT, unprecedented
spectral efficiencies of the order of 1000 bit/s/Hz per sector are achieved under various fairness criteria and pre-beamforming techniques.

We also considered the problem of downlink channel estimation and provided formulas for the asymptotic ``deterministic equivalent''
approximation of the achievable receiver SINR, which allows efficient calculation of the system performance without resorting to lengthy
Monte Carlo simulation. For a realistic SNR range around 20 dB, the effect of noisy CSIT can be quantified in $\approx 30$\% loss with respect to
the ideal CSIT case. Hence, spectral efficiencies of $\approx 700$ bit/s/Hz can be expected for the massive 3D JSDM system scenario.

The design of a JSDM system involves many choices: effective rank
$r^\star_g$ of the channel covariance matrix for each group,
pre-beamforming dimension $b_g$, number of users (downlink streams)
for each group $S_g$, for given pre-beamforming design, operating
SNR, and MU-MIMO precoding scheme. In the case of 3D beamforming,
this optimization is significantly more complicated since it has to
be repeated for groups of annular regions served simultaneously by
the vertical beamforming. One of the main merits of this paper is to
provide simple and solid design criteria for such a system, based on
the insight gained by the asymptotic analysis. In fact, a
brute-force search over the whole parameter space becomes quickly
infeasible for practical system scenarios.

We conclude this work by pointing out two interesting related topics, which are left for future work: 1) user group formation; 2)
estimation of the channel covariance matrix dominant eigenspace. User group formation considers clustering algorithms that, given $K$ users each of which is
characterized by its channel covariance dominant eigenspace, forms groups of users that can be served simultaneously using JSDM, such that
the system spectral efficiency is maximized. In order to enable user group formation and JSDM, the dominant eigenspace of each user must
be estimated from noisy samples of the received signal. Here, the problem is that for a large number of BS antennas the channel covariance matrix
is high-dimensional, and the dimension is typically comparable with the number of samples. Hence, the common wisdom on ``sample covariance'' estimation
does not apply, and more sophisticated techniques must be used (e.g., \cite[Ch. 17]{couillet2011random}, \cite{mestre2008improved,marzetta2011random}).

\appendices

\section{Deterministic equivalents for the SINR with PGP and noisy CSIT} \label{sec:determ-equiv-nonideal-csi}

We provide fixed-point equations for the calculation of the
deterministic equivalent approximations of the SINR for JSDM with
PGP, noisy CSIT and the two types of linear precoding considered in
this paper, namely, RZFBF and ZFBF. Notice that these expressions
hold for arbitrary pre-beamforming matrices, as long as they are
fixed constants independent of the instantaneous channel matrix
realizations. In particular, they hold for (approximated) BD and DFT
pre-beamforming. We consider the general case of group parameters
$\{S_g\}$, $\{b_g\}$, with equal power per stream, $P_{g_k} =
\frac{P}{S}$ for all $g_k$. The formulas below are a direct
application of the results in \cite{debbah2012}. Their derivation is
lengthy but somehow straightforward after realizing that all the
assumption in \cite{debbah2012} apply to our case. In the spirit of
striking a good balance between usefulness, conciseness and
completeness, we report the formulas without the details of their
derivation.

\subsection{Regularized Zero Forcing Precoding} \label{sec:determ-equiv-rzf}

For users in group $g$, the regularized zero forcing precoding matrix is given by
\begin{equation}
\label{rzfbf-jsdm-sect-est}
\Pm_{g,{\rm rzf}} = \bar{\zeta}_g
(\widehat{\textsf{\pH}}_g \widehat{\textsf{\pH}}_g^\herm + b_g \alpha
\Id_{b_g})^{-1} \widehat{\textsf{\pH}}_g,
\end{equation}
where $\widehat{\textsf{\pH}}_g$ is the matrix formed by the channel estimates $\widehat{\textsf{\hv}}_{g_k}$ obtained as in
(\ref{messy-mmse-estimator}).
The power normalization factor $\bar{\zeta}_g$ is given by
\begin{equation} \label{power-factor}
\bar{\zeta}_g^2 = \frac{S_g}{\trace\left( \Pm_{g,{\rm rzf}}^\herm
\Bm_g^\herm \Bm_g \Pm_{g,{\rm rzf}}\right )}
\end{equation}
Letting $\widehat{\bar{\Km}}_g = (\widehat{\textsf{\pH}}_g
\widehat{\textsf{\pH}}_g^\herm + b_g \alpha \Id_{b_g})^{-1}$, the
SINR of user $g_k$ is given by
\begin{equation}
\label{sinr-jsdm-sect-rzf-1-est} \widehat{\gamma}_{g_k,{\rm
pgp,icsi}} = \frac{\frac{P}{S} \bar{\zeta}_g^2
|\widehat{\textsf{\hv}}_{g_k}^\herm \widehat{\bar{\Km}}_g
\widehat{\textsf{\hv}}_{g_k}|^2} {\frac{P}{S} \bar{\zeta}_g^2
|\widehat{\textsf{\ev}}_{g_k}^\herm \widehat{\bar{\Km}}_g
\widehat{\textsf{\hv}}_{g_k}|^2 + \sum_{j \neq k} \frac{P}{S}
\bar{\zeta}_g^2 |\hv_{g_k}^\herm \Bm_g \widehat{\bar{\Km}}_g
\widehat{\textsf{\hv}}_{g_j}|^2 + \sum_{g' \neq g, j} \frac{P}{S}
\bar{\zeta}_{g'}^2 |\hv_{g_k}^\herm \Bm_{g'}
\widehat{\bar{\Km}}_{g'} \widehat{\textsf{\hv}}_{g'_j}|^2 + 1}
\end{equation}
where ``csi'' denotes noisy CSIT. The deterministic equivalent of
the SINR in this case is given by
\begin{equation}
\widehat{\gamma}_{g_k,{\rm pgp,rzf,csi}} - \widehat{\gamma}_{g_k,{\rm
pgp,rzf,csi}}^o \stackrel{M \rightarrow \infty}{\longrightarrow} 0
\end{equation}
with
\begin{equation} \label{ansuman-1}
\widehat{\gamma}_{g_k,{\rm pgp,rzf,csi}}^o = \frac{\frac{P}{S}
\widehat{\bar{\zeta}}_g^2 (\widehat{\bar{m}}_{g}^o)^2}
{\frac{P}{S} \widehat{\bar{\zeta}}_g^2 \widehat{\bar{E}}_{g}^o + \frac{P}{S}
\widehat{\bar{\zeta}}_g^2 \widehat{\bar{\Upsilon}}_{g,g}^o +
\left (1 + \sum_{g'
\neq g} \frac{P}{S} \widehat{\bar{\zeta}}_{g'}^2
\widehat{\bar{\Upsilon}}_{g,g'}^o \right ) \left (1 + \widehat{\bar{m}}_{g}^o \right )^2}
\end{equation}
where $\widehat{\bar{\zeta}}_g^2 = \frac{1}{\widehat{\bar{\Gamma}}_g^o}$ and
the quantities $\widehat{\bar{m}}_{g}^o$,
$\widehat{\bar{\Upsilon}}_{g,g}^o$, $\widehat{\bar{\Upsilon}}_{g,g'}^o$ and
$\widehat{\bar{\Gamma}}_g^o$ are given by
\begin{eqnarray}
\label{fixed-pt-1-rzfbf-jsdm-sect-est} \widehat{\bar{m}}_{g}^o &=& \frac{1}{b_g} \trace \left ( \widehat{\bar{\Rm}}_{g} \widehat{\bar{\Tm}}_g \right )\\
\label{fixed-pt-2-rzfbf-jsdm-sect-est} \widehat{\bar{\Tm}}_g &=& \left(
\frac{S_g}{b_g}  \frac{\widehat{\bar{\Rm}}_{g}}{1 + \widehat{\bar{m}}_{g}^o}
+
\alpha \Id_{b_g}\right)^{-1}\\
\widehat{\bar{\Gamma}}_g^o &=& \frac{1}{b_g} \frac{\widehat{\bar{n}}_{g}}{(1
+
\widehat{\bar{m}}_{g}^o)^2}\\
\widehat{\bar{n}}_{g} &=& \frac{\frac{1}{b_g} \trace \left (
\widehat{\bar{\Rm}}_{g} \widehat{\bar{\Tm}}_g \Bm_{g}^\herm\Bm_{g}
\widehat{\bar{\Tm}}_g \right )}{1 - \frac{\frac{S_g}{b_g} \trace \left (
\widehat{\bar{\Rm}}_{g}
\widehat{\bar{\Tm}}_g \widehat{\bar{\Rm}}_{g} \widehat{\bar{\Tm}}_g  \right )}{b_g (1 + \widehat{\bar{m}}_{g}^o)^2}} \\
\widehat{\bar{E}}_{g}^o &=& \frac{1}{b_g} \frac{\frac{1}{b_g} \trace
\left ( \widehat{\bar{\Rm}}_{g} \widehat{\bar{\Tm}}_g (\bar{\Rm}_g -
\widehat{\bar{\Rm}}_g) \widehat{\bar{\Tm}}_g \right )}{1 -
\frac{\frac{S_g}{b_g} \trace \left ( \widehat{\bar{\Rm}}_{g}
\widehat{\bar{\Tm}}_g \widehat{\bar{\Rm}}_{g} \widehat{\bar{\Tm}}_g  \right )}{b_g (1 + \widehat{\bar{m}}_{g}^o)^2}}\\
\widehat{\bar{\Upsilon}}_{g,g}^o &=& (1 + \widehat{\bar{m}}_{g}^o)^2 A_1 - \left[ 2 \widehat{\bar{m}}_{g}^o(1 + \widehat{\bar{m}}_{g}^o) - (\widehat{\bar{m}}_{g}^o)^2 \right] A_2\\
A_1 &=& \frac{1}{b_g} (S_g - 1)
 \frac{\widehat{\bar{n}}_{g,g,1}}{(1 + \widehat{\bar{m}}_{g}^o)^2}\\
A_2 &=& \frac{1}{b_g} (S_g - 1)
 \frac{\widehat{\bar{n}}_{g,g,2}}{(1 + \widehat{\bar{m}}_{g}^o)^2}\\
\widehat{\bar{n}}_{g,g,1} &=& \frac{\frac{1}{b_g} \trace \left (
\widehat{\bar{\Rm}}_{g} \widehat{\bar{\Tm}}_g \bar{\Rm}_{g}
\widehat{\bar{\Tm}}_g \right )}{1 - \frac{\frac{S_g}{b_g} \trace\left (
\widehat{\bar{\Rm}}_{g} \widehat{\bar{\Tm}}_g \widehat{\bar{\Rm}}_{g}
\widehat{\bar{\Tm}}_g \right )}{b_g (1 +
\widehat{\bar{m}}_{g}^o)^2}}
\end{eqnarray}
\begin{eqnarray}
\widehat{\bar{n}}_{g,g,2} &=& \frac{\frac{1}{b_g} \trace \left (
\widehat{\bar{\Rm}}_{g} \widehat{\bar{\Tm}}_g \widehat{\bar{\Rm}}_{g}
\widehat{\bar{\Tm}}_g \right )}{1 - \frac{\frac{S_g}{b_g} \trace\left (
\widehat{\bar{\Rm}}_{g} \widehat{\bar{\Tm}}_g \widehat{\bar{\Rm}}_{g}
\widehat{\bar{\Tm}}_g \right )}{b_g (1 +
\widehat{\bar{m}}_{g}^o)^2}}\\
\widehat{\bar{\Upsilon}}_{g,g'}^o &=& \frac{S_{g'}}{b_{g'}} \frac{
\widehat{\bar{n}}_{g',g}}{(1 + \widehat{\bar{m}}_{g'}^o)^2}\\
\widehat{\bar{n}}_{g',g} &=& \frac{\frac{1}{b_{g'}} \trace \left (
\widehat{\bar{\Rm}}_{g'} \widehat{\bar{\Tm}}_{g'}
\Bm_{g'}^\herm\Rm_{g}\Bm_{g'} \widehat{\bar{\Tm}}_{g'}  \right )}{1 -
\frac{\frac{S_{g'}}{b_{g'}} \trace \left ( \widehat{\bar{\Rm}}_{g'}
\widehat{\bar{\Tm}}_{g'} \widehat{\bar{\Rm}}_{g'} \widehat{\bar{\Tm}}_{g'}
\right )}{b_{g'} (1 + \widehat{\bar{m}}_{g'}^o)^2}}
\end{eqnarray}

\subsection{Zero Forcing Precoding} \label{sec:determ-equiv-zf}

For $\alpha = 0$, the precoding matrix in
(\ref{rzfbf-jsdm-sect-est}) reduces to the zero forcing precoding
matrix given by
\begin{equation}
\Pm_{g,{\rm zf}} = \bar{\zeta}_g
\widehat{\textsf{\pH}}_g(\widehat{\textsf{\pH}}_g^\herm
\widehat{\textsf{\pH}}_g)^{-1}
\end{equation}
where $\bar{\zeta}_g$ is the power normalization
factor given by
\begin{equation}
\bar{\zeta}_g^2 = \frac{S_g}{{\rm tr} (\Pm_{g,{\rm zf}}^\herm \Bm_g \Bm_g^\herm \Pm_{g,{\rm zf}})}
\end{equation}
Letting $\widehat{\bar{\Km}}_g = \widehat{\textsf{\pH}}_g
(\widehat{\textsf{\pH}}_g^\herm \widehat{\textsf{\pH}}_g)^{-2}
\widehat{\textsf{\pH}}_g^\herm$, the SINR of user $g_k$ is given by
\begin{equation}
\label{sinr-zfbf-sect-jsdm-1-est}
\widehat{\gamma}_{g_k,{\rm pgp,zf,csi}} = \frac{\frac{P}{S} \widehat{\bar{\zeta}}_g^2
|\widehat{\textsf{\hv}}_{g_k}^\herm \widehat{\bar{\Km}}_g
\widehat{\textsf{\hv}}_{g_k} |^2}
{\frac{P}{S} \widehat{\bar{\zeta}}_g^2 |\widehat{\textsf{\ev}}_{g_k}^\herm
\widehat{\bar{\Km}}_g \widehat{\textsf{\hv}}_{g_k} |^2 + \sum_{j \neq
k} \frac{P}{S} \widehat{\bar{\zeta}}_g^2 |\hv_{g_k}^\herm
\Bm_g \widehat{\bar{\Km}}_g \widehat{\textsf{\hv}}_{g_j} |^2 +
\sum_{g' \neq g,j} \frac{P}{S} \widehat{\bar{\zeta}}_{g'}^2
|\hv_{g_k}^\herm \Bm_{g'} \widehat{\bar{\Km}}_{g'}
\widehat{\textsf{\hv}}_{g'_j} |^2 + 1}
\end{equation}
The deterministic equivalent of the SINR is given as
\begin{equation}
\widehat{\gamma}_{g_k,{\rm pgp,zf,csi}} - \widehat{\gamma}_{g_k,{\rm pgp,zf,csi}}^o \stackrel{M \rightarrow
\infty}{\longrightarrow} 0
\end{equation}
where $\widehat{\gamma}_{g_k,{\rm pgp,zf,csi}}^o$ is given by
\begin{eqnarray}  \label{ansuman-2}
\widehat{\gamma}_{g_k,{\rm pgp,zf,csi}}^o &=& \frac{\frac{P}{S}
\widehat{\bar{\zeta}}_g^2 } {1 + \frac{P}{S}
\widehat{\bar{\zeta}}_g^2
\frac{\widehat{\bar{E}}_{g}^o}{(\widehat{\bar{m}}_{g}^o)^2} +
\frac{P}{S} \widehat{\bar{\zeta}}_{g}^2
\widehat{\bar{\Upsilon}}_{g,g}^o + \sum_{g' \neq g} \frac{P}{S}
\widehat{\bar{\zeta}}_{g'}^2
\widehat{\bar{\Upsilon}}_{g,g'}^o}\nonumber\\
&=& \frac{\frac{P}{S} \widehat{\bar{\zeta}}_g^2 } {1 + \frac{P}{S}
\widehat{\bar{\zeta}}_g^2 S_g
\frac{\widehat{\bar{E}}_{g}^o}{(\widehat{\bar{m}}_{g}^o)^2} +
\sum_{g' \neq g} \frac{P}{S} \widehat{\bar{\zeta}}_{g'}^2
\widehat{\bar{\Upsilon}}_{g,g'}^o}
\end{eqnarray}
with $\widehat{\bar{\zeta}}_g^2 =
\frac{1}{\widehat{\bar{\Gamma}}_g^o}$ and the quantities
$\widehat{\bar{\Gamma}}_g^o$,
$\widehat{\bar{\Upsilon}}_{g,g'}^o$, and
$\widehat{\bar{m}}_{g}^o$ are given by~\footnote{It is easy
to see that when $\Bm_g^\herm \Bm_g = \Id_{b_g}$,
$\widehat{\bar{n}}_{g} = \widehat{\bar{m}}_{g}$}
\begin{eqnarray}
\label{fixed-pt-1-rzfbf-jsdm-sect-est}
\widehat{\bar{m}}_{g}^o &=& \frac{1}{b_g} \trace \left ( \widehat{\bar{\Rm}}_{g} \widehat{\bar{\Tm}}_g \right )\\
\label{fixed-pt-2-rzfbf-jsdm-sect-est} \widehat{\bar{\Tm}}_g
&=& \left( \frac{S_g}{b_g}
\frac{\widehat{\bar{\Rm}}_{g}}{\widehat{\bar{m}}_{g}^o} +
\Id_{b_g}\right)^{-1}\\
\widehat{\bar{\Gamma}}_g^o &=& \frac{1}{b_g}
\frac{\widehat{\bar{n}}_{g}}{
(\widehat{\bar{m}}_{g}^o)^2}\\
\widehat{\bar{\Upsilon}}_{g,g'}^o &=& \frac{S_{g'}}{b_{g'}}
 \frac{\widehat{\bar{n}}_{g',g}}{(\widehat{\bar{m}}_{g'}^o)^2}\\
\widehat{\bar{n}}_{g} &=& \frac{\frac{1}{b_g} \trace \left (
\widehat{\bar{\Rm}}_{g} \widehat{\bar{\Tm}}_g
\Bm_{g}^\herm\Bm_{g} \widehat{\bar{\Tm}}_g \right )}{1 -
\frac{\frac{S_g}{b_g} \trace \left ( \widehat{\bar{\Rm}}_{g}
\widehat{\bar{\Tm}}_g \widehat{\bar{\Rm}}_{g} \widehat{\bar{\Tm}}_g  \right )}{b_g (\widehat{\bar{m}}_{g}^o)^2}}\\
\widehat{\bar{n}}_{g',g} &=& \frac{\frac{1}{b_{g'}} \trace
\left ( \widehat{\bar{\Rm}}_{g'} \widehat{\bar{\Tm}}_{g'}
\Bm_{g'}^\herm\Rm_{g}\Bm_{g'} \widehat{\bar{\Tm}}_{g'}
\right )}{1 - \frac{\frac{S_{g'}}{b_{g'}} \trace \left (
\widehat{\bar{\Rm}}_{g'} \widehat{\bar{\Tm}}_{g'}
\widehat{\bar{\Rm}}_{g'} \widehat{\bar{\Tm}}_{g'} \right
)}{b_{g'}
(\widehat{\bar{m}}_{g'}^o)^2}}\\
\widehat{\bar{\Upsilon}}_{g,g}^o &=& A_1 - A_2\\
A_1 &=& \frac{1}{b_g} (S_g - 1)
 \frac{\widehat{\bar{n}}_{g,g,1}}{(\widehat{\bar{m}}_{g}^o)^2}\\
A_2 &=& \frac{1}{b_g} (S_g - 1) \frac{\widehat{\bar{n}}_{g,g,2}}{(\widehat{\bar{m}}_{g}^o)^2}\\
\widehat{\bar{n}}_{g,g,1} &=& \frac{\frac{1}{b_g} \trace
\left ( \widehat{\bar{\Rm}}_{g} \widehat{\bar{\Tm}}_g
\bar{\Rm}_{g} \widehat{\bar{\Tm}}_g \right )}{1 -
\frac{\frac{S_g}{b_g} \trace\left ( \widehat{\bar{\Rm}}_{g}
\widehat{\bar{\Tm}}_g \widehat{\bar{\Rm}}_{g}
\widehat{\bar{\Tm}}_g \right )}{b_g (\widehat{\bar{m}}_{g}^o)^2}}\\
\widehat{\bar{n}}_{g,g,2} &=& \frac{\frac{1}{b_g} \trace
\left ( \widehat{\bar{\Rm}}_{g} \widehat{\bar{\Tm}}_g
\widehat{\bar{\Rm}}_{g} \widehat{\bar{\Tm}}_g \right )}{1 -
\frac{\frac{S_g}{b_g} \trace\left ( \widehat{\bar{\Rm}}_{g}
\widehat{\bar{\Tm}}_g \widehat{\bar{\Rm}}_{g}
\widehat{\bar{\Tm}}_g \right )}{b_g (\widehat{\bar{m}}_{g}^o)^2}}\\
\widehat{\bar{E}}_{g}^o &=& \frac{1}{b_g}
\frac{\frac{1}{b_g} \trace \left ( \widehat{\bar{\Rm}}_{g}
\widehat{\bar{\Tm}}_g (\bar{\Rm}_g - \widehat{\bar{\Rm}}_g)
\widehat{\bar{\Tm}}_g \right )}{1 - \frac{\frac{S_g}{b_g}
\trace \left ( \widehat{\bar{\Rm}}_{g} \widehat{\bar{\Tm}}_g
\widehat{\bar{\Rm}}_{g} \widehat{\bar{\Tm}}_g  \right )}{b_g
(\widehat{\bar{m}}_{g}^o)^2}}
\end{eqnarray}
In order to obtain the desired expression (\ref{ansuman-2}), we notice that
\begin{equation}
\frac{\widehat{\bar{E}}_{g}^o}{(\widehat{\bar{m}}_{g}^o)^2}
+ \widehat{\bar{\Upsilon}}_{g,g}^o =
S_g\frac{\widehat{\bar{E}}_{g}^o}{(\widehat{\bar{m}}_{g}^o)^2}
\end{equation}

\section{General formula for $S(\xi)$}
\label{bessel-derive}

We find the general expression for $S(\xi)$ defined in (\ref{PSD}) for $r_m = [\Rm]_{\ell,\ell-m}$ with $[\Rm]_{m,p}$ given by (\ref{correlation-matrix-ULA}),
without any restriction on the AoA range. We have:
\begin{eqnarray} \label{eqn:s-xi}
S(\xi) & = & \sum_{m=-\infty}^\infty \left [ \frac{1}{2\Delta} \int_{-\Delta+\theta}^{\Delta+\theta}   e^{-j2\pi D m \sin(\alpha)} d\alpha \right ] e^{-j2\pi \xi m} \nonumber \\
& = & \frac{1}{2\Delta} \int_{-\Delta+\theta}^{\Delta+\theta} \left [ \sum_{m=-\infty}^\infty e^{-j2\pi m (D \sin(\alpha) + \xi)} \right ] d\alpha \nonumber \\
& = & \frac{1}{2\Delta} \int_{-\Delta+\theta}^{\Delta+\theta} \left [ \sum_{m=-\infty}^\infty \delta (D \sin(\alpha) + \xi - m) \right ] d\alpha \nonumber \\
& = & \frac{1}{2\Delta} \int \left [ \sum_{m=-\infty}^\infty \delta
(z + \xi - m) \right ] \frac{dz}{\sqrt{D^2 - z^2}},
\end{eqnarray}

The limits in (\ref{eqn:s-xi}) depend on the range of
$[\theta-\Delta,\theta+\Delta]$. We distinguish the following cases:
\begin{enumerate}

\item For $\theta+\Delta < -\frac{\pi}{2}, \theta-\Delta > \frac{\pi}{2}$ and $-\frac{\pi}{2} \leq \theta-\Delta < \theta+\Delta \leq \frac{\pi}{2}$, (\ref{eqn:s-xi}) becomes
\begin{equation}
\sum_{m=-\infty}^\infty \left [ \int_{\min (D \sin(\theta-\Delta),D
\sin(\theta+\Delta))}^{\max (D \sin(\theta-\Delta),D
\sin(\theta+\Delta))}
\delta (z + \xi - m) \frac{dz}{\sqrt{D^2 - z^2}} \right ]
\end{equation}

\item For $\theta-\Delta < -\frac{\pi}{2}, \theta+\Delta > \frac{\pi}{2}$, (\ref{eqn:s-xi}) becomes
\begin{align}
& \sum_{m=-\infty}^\infty \left [ \int_{-D}^{D \sin(\theta-\Delta)}  \delta (z + \xi - m) \frac{dz}{\sqrt{D^2 - z^2}} + \int_{-D}^{D} \delta (z + \xi - m) \frac{dz}{\sqrt{D^2 - z^2}} \right . \\
& + \left . \int_{D \sin(\theta+\Delta)}^{D}  \delta (z + \xi - m) \frac{dz}{\sqrt{D^2 - z^2}} \right ]
\end{align}

\item For $\theta-\Delta < -\frac{\pi}{2}, -\frac{\pi}{2} \leq \theta+\Delta \leq \frac{\pi}{2}$, (\ref{eqn:s-xi}) becomes
\begin{equation}
\sum_{m=-\infty}^\infty \left[ \int_{-D}^{D \sin(\theta-\Delta)}
 \delta (z + \xi - m) \frac{dz}{\sqrt{D^2 - z^2}} +
\int_{-D}^{D \sin(\theta+\Delta)}
 \delta (z + \xi - m) \frac{dz}{\sqrt{D^2 - z^2}} \right ]
\end{equation}

\item For $-\frac{\pi}{2} \leq \theta-\Delta \leq \frac{\pi}{2}, \theta+\Delta > \frac{\pi}{2}$, (\ref{eqn:s-xi}) becomes
\begin{equation}
\sum_{m=-\infty}^\infty \left[ \int_{D \sin(\theta-\Delta)}^{D}
\delta (z + \xi - m) \frac{dz}{\sqrt{D^2 - z^2}} + \int_{D
\sin(\theta+\Delta)}^{D}
 \delta (z + \xi - m) \frac{dz}{\sqrt{D^2 - z^2}} \right ]
\end{equation}
\end{enumerate}
Now, owing to the property of the Dirac delta function, we have
\begin{equation}
\sum_{m=-\infty}^\infty \left [ \int_{A}^{B} \delta (z + \xi - m)
\frac{dz}{\sqrt{D^2 - z^2}} \right ] = \sum_{m \in [A + \xi, B +
\xi]} \frac{1}{\sqrt{D^2 - (m - \xi)^2}}
\end{equation}
as a result of which we can write $S(\xi)$ for the cases identified above as
\begin{enumerate}

\item Case $\theta+\Delta < -\frac{\pi}{2}, \theta-\Delta >
\frac{\pi}{2}$ and $-\frac{\pi}{2} \leq \theta-\Delta <
\theta+\Delta \leq \frac{\pi}{2}$
\begin{equation}
S(\xi) = \frac{1}{2\Delta} \sum_{m \in
[\min(D\sin(-\Delta+\theta),D\sin(\Delta+\theta)) + \xi,
\max(D\sin(-\Delta+\theta),D\sin(\Delta+\theta)) + \xi]}
\frac{1}{\sqrt{D^2 - (m - \xi)^2}}
\end{equation}

\item Case $\theta-\Delta < -\frac{\pi}{2}, \theta+\Delta >
\frac{\pi}{2}$
\begin{eqnarray}
S(\xi) &=& \frac{1}{2\Delta} \sum_{m \in [-D + \xi,
D\sin(-\Delta+\theta) + \xi]} \frac{1}{\sqrt{D^2 - (m - \xi)^2}} +
\frac{1}{2\Delta} \sum_{m \in (-D + \xi, D + \xi)}
\frac{1}{\sqrt{D^2 - (m - \xi)^2}}
\nonumber\\
&& + \frac{1}{2\Delta} \sum_{m \in [D\sin(\Delta+\theta) + \xi,D +
\xi]} \frac{1}{\sqrt{D^2 - (m - \xi)^2}} \nonumber
\end{eqnarray}

\item Case $\theta-\Delta < -\frac{\pi}{2}, -\frac{\pi}{2} \leq \theta+\Delta
\leq \frac{\pi}{2}$
\begin{equation}
S(\xi) = \frac{1}{2\Delta} \sum_{m \in [-D + \xi,
D\sin(-\Delta+\theta) + \xi]} \frac{1}{\sqrt{D^2 - (m - \xi)^2}} +
\frac{1}{2\Delta} \sum_{m \in (-D + \xi, D\sin(\Delta+\theta) +
\xi]} \frac{1}{\sqrt{D^2 - (m - \xi)^2}}
\end{equation}

\item Case $-\frac{\pi}{2} \leq \theta-\Delta \leq \frac{\pi}{2}, \theta+\Delta >
\frac{\pi}{2}$
\begin{equation}
S(\xi) = \frac{1}{2\Delta} \sum_{m \in [D\sin(-\Delta+\theta) +
\xi,D + \xi]} \frac{1}{\sqrt{D^2 - (m - \xi)^2}} + \frac{1}{2\Delta}
\sum_{m \in [D\sin(\Delta+\theta) + \xi,D + \xi)} \frac{1}{\sqrt{D^2
- (m - \xi)^2}}
\end{equation}
\end{enumerate}
It is easy to see that the formula reduces to (\ref{fikissimo}) when $-\frac{\pi}{2} \leq \theta-\Delta < \theta+\Delta \leq
\frac{\pi}{2}$. Taking the limits from $-\pi$ to $\pi$ recovers the Fourier transform of the Bessel $J_0$ function commonly used to model
correlated Rayleigh fading in an isotropic scattering environment \cite{bello1963characterization},
given by   $\frac{1}{\pi} \frac{{\rm rect}(\xi/2D)}{\sqrt{D^2 - \xi^2}}, \xi \in [-1/2,1/2]$ for $D \in [0,\frac{1}{2}]$.

\bibliographystyle{IEEEtran}
\bibliography{jsdm_ciss}

\end{document}